\documentclass[journal]{IEEEtran}

\usepackage{cite}
\usepackage[pdftex]{graphicx}
\makeatletter
\let\MYcaption\@makecaption
\makeatother

\usepackage[font=footnotesize]{subcaption}

\makeatletter
\let\@makecaption\MYcaption
\makeatother
\usepackage{amsmath,amssymb,amsfonts}
\usepackage{amsthm}

\newtheorem{asmp}{Assumption}

\makeatletter
\newcounter{parentcond}

\makeatother

\newtheorem{lem}{Lemma}
\theoremstyle{definition}
\newtheorem{rem}{Remark}

\usepackage[ruled,norelsize]{algorithm2e}
\SetKwBlock{Repeat}{repeat}{}
\usepackage{array}
\usepackage{bbm}
\usepackage{comment}
\usepackage{hhline}
\usepackage[flushleft]{threeparttable}
\usepackage{nccmath}
\usepackage{multicol}
\usepackage{makecell}
\usepackage{tabularx}
\usepackage[hidelinks]{hyperref}
\usepackage{xcolor}

\newcolumntype{Y}{>{\centering\arraybackslash}X}
\newcolumntype{b}{>{\hsize=1.17\hsize}Y}
\newcolumntype{z}{>{\hsize=1.28\hsize}Y}
\newcolumntype{s}{>{\hsize=.55\hsize}Y}
\newcolumntype{y}{>{\hsize=.75\hsize}Y}
\newcolumntype{x}{>{\hsize=1.25\hsize}Y}

\makeatletter
\let\save@mathaccent\mathaccent
\newcommand*\if@single[3]{%
  \setbox0\hbox{${\mathaccent"0362{#1}}^H$}%
  \setbox2\hbox{${\mathaccent"0362{\kern0pt#1}}^H$}%
  \ifdim\ht0=\ht2 #3\else #2\fi
  }
\newcommand*\rel@kern[1]{\kern#1\dimexpr\macc@kerna}
\newcommand*\widebar[1]{\@ifnextchar^{{\wide@bar{#1}{0}}}{\wide@bar{#1}{1}}}
\newcommand*\wide@bar[2]{\if@single{#1}{\wide@bar@{#1}{#2}{1}}{\wide@bar@{#1}{#2}{2}}}
\newcommand*\wide@bar@[3]{%
  \begingroup
  \def\mathaccent##1##2{%
    \let\mathaccent\save@mathaccent
    \if#32 \let\macc@nucleus\first@char \fi
    \setbox\z@\hbox{$\macc@style{\macc@nucleus}_{}$}%
    \setbox\tw@\hbox{$\macc@style{\macc@nucleus}{}_{}$}%
    \dimen@\wd\tw@
    \advance\dimen@-\wd\z@
    \divide\dimen@ 3
    \@tempdima\wd\tw@
    \advance\@tempdima-\scriptspace
    \divide\@tempdima 10
    \advance\dimen@-\@tempdima
    \ifdim\dimen@>\z@ \dimen@0pt\fi
    \rel@kern{0.6}\kern-\dimen@
    \if#31
      \overline{\rel@kern{-0.6}\kern\dimen@\macc@nucleus\rel@kern{0.4}\kern\dimen@}%
      \advance\dimen@0.4\dimexpr\macc@kerna
      \let\final@kern#2%
      \ifdim\dimen@<\z@ \let\final@kern1\fi
      \if\final@kern1 \kern-\dimen@\fi
    \else
      \overline{\rel@kern{-0.6}\kern\dimen@#1}%
    \fi
  }%
  \macc@depth\@ne
  \let\math@bgroup\@empty \let\math@egroup\macc@set@skewchar
  \mathsurround\z@ \frozen@everymath{\mathgroup\macc@group\relax}%
  \macc@set@skewchar\relax
  \let\mathaccentV\macc@nested@a
  \if#31
    \macc@nested@a\relax111{#1}%
  \else
    \def\gobble@till@marker##1\endmarker{}%
    \futurelet\first@char\gobble@till@marker#1\endmarker
    \ifcat\noexpand\first@char A\else
      \def\first@char{}%
    \fi
    \macc@nested@a\relax111{\first@char}%
  \fi
  \endgroup
}
\makeatother

\makeatletter
\newcommand*{\inlineequation}[2][]{%
  \begingroup
    \refstepcounter{equation}%
    \ifx\\#1\\%
    \else
      \label{#1}%
    \fi
    \relpenalty=10000 %
    \binoppenalty=10000 %
    \ensuremath{%
      #2%
    }%
    ~\@eqnnum
  \endgroup
}
\makeatother
\makeatletter
\newcommand{\removelatexerror}{\let\@latex@error\@gobble}
\makeatother

\SetKwFor{At}{at}{do}{end}

\begin{document}

\setlength{\textfloatsep}{5pt plus 1.0pt minus 2.0pt}

\title{A Unified Algorithmic Framework for Distributed Adaptive Signal and Feature Fusion Problems \\\LARGE --- Part I: Algorithm Derivation}

\author{Cem Ates~Musluoglu,
        and~Alexander~Bertrand,~\IEEEmembership{Senior~Member,~IEEE}
\thanks{Copyright \copyright 2023 IEEE. Personal use of this material is permitted. Permission from IEEE must be obtained for all other uses, in any current or future media, including reprinting/republishing this material for advertising or promotional purposes, creating new collective works, for resale or redistribution to servers or lists, or reuse of any copyrighted component of this work in other works.}
\thanks{This project has received funding from the European Research Council (ERC) under the European Union's Horizon 2020 research and innovation programme (grant agreement No 802895). The authors also acknowledge the financial support of the FWO (Research Foundation Flanders) for project G081722N, and the Flemish Government (AI Research Program).}
\thanks{C.A. Musluoglu and A. Bertrand are with KU Leuven, Department of Electrical Engineering (ESAT), Stadius Center for Dynamical Systems, Signal Processing and Data Analytics, Kasteelpark Arenberg 10, box 2446, 3001 Leuven, Belgium and with Leuven.AI - KU Leuven institute for AI. e-mail: cemates.musluoglu, alexander.bertrand @esat.kuleuven.be}
\thanks{A companion paper submitted together with this paper is provided in \cite{musluoglu2022unifiedp2}.}
\thanks{Digital Object Identifier: 10.1109/TSP.2023.3275272}
}

\maketitle

\begin{abstract}
In this paper, we describe a general algorithmic framework for solving linear signal or feature fusion optimization problems in a distributed setting, for example in a wireless sensor network (WSN). These problems require linearly combining the observed signals (or features thereof) collected at the various sensor nodes to satisfy a pre-defined optimization criterion. The framework covers several classical spatial filtering problems, including minimum variance beamformers, multi-channel Wiener filters, principal component analysis, canonical correlation analysis, (generalized) eigenvalue problems, etc. The proposed distributed adaptive signal fusion (DASF) algorithm is an iterative method that solves these types of problems by allowing each node to share a linearly compressed version of the local sensor signal observations with its neighbors to reduce the energy and bandwidth requirements of the network. We first discuss the case of fully-connected networks and then extend the analysis to more general network topologies. The general DASF algorithm is shown to have several existing distributed algorithms from the literature as a special case, while at the same time allowing to solve new distributed problems as well with guaranteed convergence and optimality. This paper focuses on the algorithm derivation of the DASF framework along with simulations demonstrating its performance. A technical analysis along with convergence conditions and proofs are provided in a companion paper.
\end{abstract}


\section{Introduction}

\IEEEPARstart{W}{ireless} Sensor Networks (WSNs) consist of a wireless network of sensor nodes, which collect, process, and share data in order to solve a specific signal processing task in a collaborative fashion. Such WSNs allow to easily acquire data at multiple locations simultaneously, which is useful in several application domains including health monitoring \cite{latre2011survey,bertrand2015distributed}, acoustics \cite{bertrand2011applications,markovich2015optimal}, structural monitoring \cite{xu2004wireless}, environmental studies \cite{othman2012wireless,albaladejo2010wireless}, and many others \cite{akyildiz2002wireless,yick2008wireless}.

In many cases, the aim is to find or estimate a common signal, filter or a set of parameters that satisfy a pre-defined optimality criterion involving the observation data from all the nodes \cite{sayed2014adaptation,bertsekas1989parallel}. Various ``work horse'' strategies and frameworks have been described previously to solve such problems in a distributed fashion. Well-known examples are consensus \cite{olfati2005consensus,olfati2007consensus}, incremental strategies \cite{bertsekas1997new,lopes2007incremental}, diffusion \cite{lopes2008diffusion,cattivelli2009diffusion,chen2012diffusion}, gossip \cite{boyd2006randomized,dimakis2010gossip}, or the alternating direction method of multipliers \cite{boyd2011distributed,wei2012distributed}.

\begin{figure}[t]
  \includegraphics[width=0.48\textwidth]{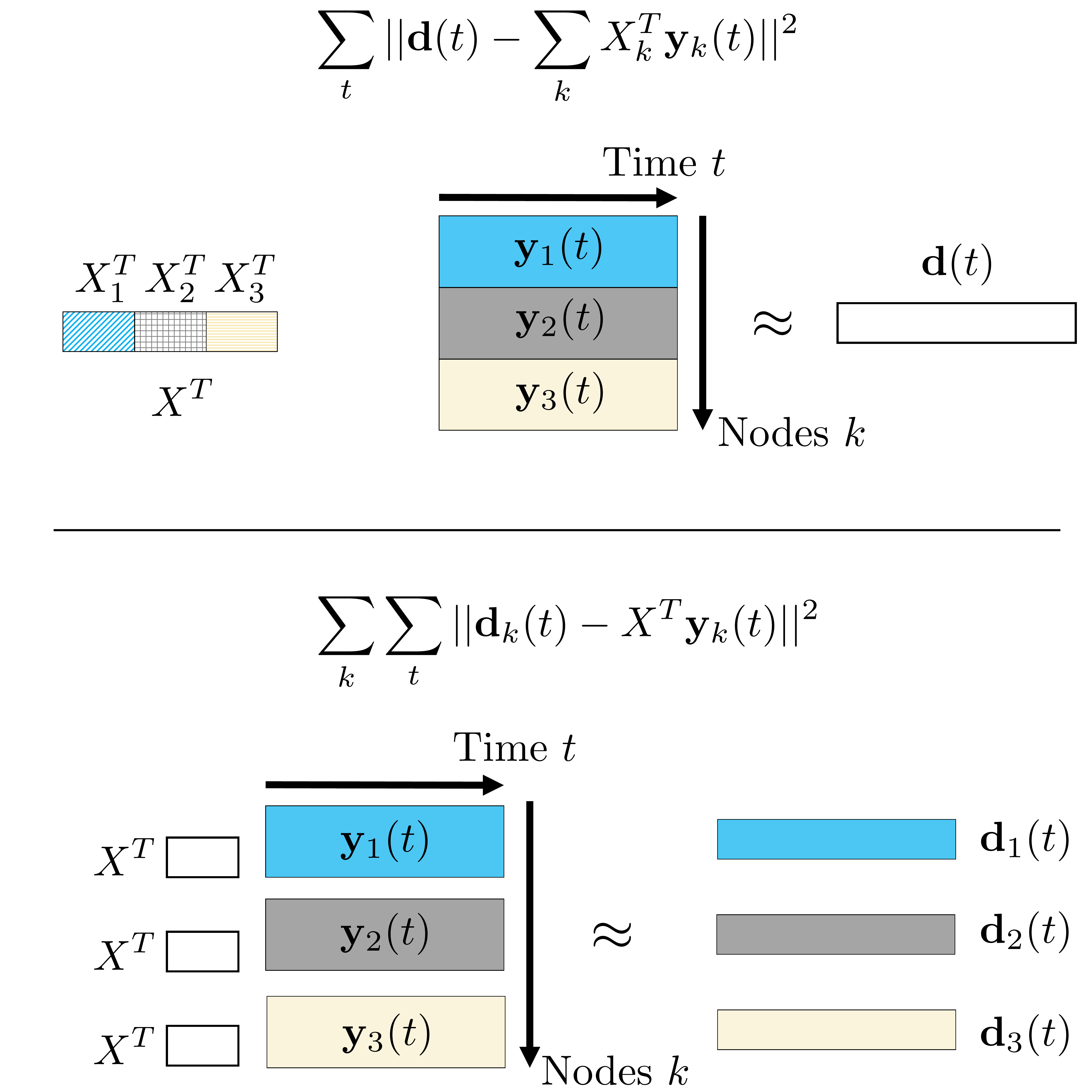}
  \caption{Comparison of the $3-$node problem setting of the DSFO (top) and the traditional consensus-type setting (bottom) with corresponding example objective functions for the case of least squares estimation. In the consensus setting, note that the objective is per-node separable and has a shared optimization variable $X$ which is assumed to be the same across all nodes (which is why there is no node subscript $k$).}
  \label{fig:data_comp}
\end{figure}

In order to achieve a distributed implementation, most of these methods rely on the separability of the global cost function $f$ as a sum of local functions $f_k$: $f(X)=\sum_k f_k(X)$, where $f_k$ depends only on the local data of node $k$, and where $X$ is a \textit{shared} optimization variable across all nodes. In this case, we say that $f$ is per-node separable. However, there exist various cases where this property is not satisfied, e.g., when optimizing a spatial filter $X$ that linearly combines the signals from different nodes. Classical examples are adaptive beamformers \cite{gershman2010convex,vorobyov2013principles}, multi-channel Wiener filtering \cite{doclo2002gsvd,somers2018generic}, principal component analysis, filters based on (generalized) eigenvectors of spatial covariance matrices \cite{belhumeur1997eigenfaces,sugiyama2007dimensionality}, etc. In these problems, the aim is to find a network-wide linear spatial filter $X\in\mathbb{R}^{M\times Q}$, $Q<M$ to apply to the network-wide $M-$channel time signal $\mathbf{y}(t)\in\mathbb{R}^{M}$ containing all sensor channels\footnote{In this paper, we adopt the terminology from the field of sensor arrays and multi-channel signal processing although the results are also applicable in a more general context, where $\mathbf{y}$ can be viewed as a generic feature vector in an $M-$dimensional feature space where distributed agents each observe a part of the feature vector.} from all nodes in the network, where $t$ denotes the time or sample index. The filter $X$ is typically designed to exploit the spatial correlation across the different nodes to optimize some network-wide objective function in the form $f(X^T\mathbf{y}(t))$, with $T$ the transpose operator. In this case, the function $f$ itself is not per-node separable, but the argument is, i.e., $f(X^T\mathbf{y}(t))=f(\sum_k X_k^T\mathbf{y}_k(t))$. If $\mathbf{y}$ is interpreted as a feature vector, this type of separation of the observed data is also known as feature partitioning or distributed features \cite{ying2018supervised,zheng2010attribute,mota2011distributed,manss2018distributed,gratton2021distributed,zhang2018communication}. In this paper, we refer to such cases as a distributed signal fusion optimization (DSFO) problem to emphasize that the framework also applies to traditional array processing problems such as beamforming or spatial filtering. A visual example that illustrates the conceptual difference between both aforementioned types of data separation, i.e., separability of the cost function versus separability of the argument, is given in Figure \ref{fig:data_comp} for the case of least squares estimation.

A commonly encountered strategy to solve DSFO problems is to compute all inner products involving the data vector $\mathbf{y}$ via a standard consensus-type subroutine that performs in-network averaging or summation \cite{li2011distributed,zeng2013distributed}, or by artificially rewriting the problem as a consensus problem (e.g., by treating the filter output $\mathbf{d}$ itself as a shared (consensus) optimization variable in the example of Figure \ref{fig:data_comp}). However, this strategy typically results in a distributed algorithm with nested iterative loops for each sample time $t$, each in itself requiring a substantial number of communication rounds. Such an iterative distributed subroutine typically has to be executed from scratch for each new sample observation at the sensors. This leads to a large communication burden, which also scales poorly with network size, i.e., a larger network increases the \textit{per node} transmission cost in these subroutines. In many cases, the use of such consensus-type subroutines even leads to a setting where each node shares more data than what it actually collects at its sensors.

Based on the block partitioning of the DSFO problem in Figure \ref{fig:data_comp}, a tempting alternative strategy could be to use a nonlinear Gauss-Seidel method \cite{bertsekas1989parallel,grippo2000convergence}, or so-called block coordinate descent algorithms. These are iterative algorithms in which a block of variables in $X$ is optimized while keeping all others fixed, and where the fixed blocks change across iterations. By selecting the blocks of $X$ according to the nodes (i.e., the $X_k$'s in Figure \ref{fig:data_comp}), each iteration of the nonlinear Gauss-Seidel method can then be ``outsourced'' to the node that is responsible for optimizing the corresponding coordinates, which makes it a better fit for the class of problems we are interested in. However, the convergence results for such nonlinear Gauss-Seidel methods often require conditions such as convexity assumptions or constraint sets that can be written as Cartesian products, where each factor corresponds to a constraint set for the block $X_k$ \cite{bertsekas1989parallel,grippo2000convergence}, which would not be satisfied in many spatial filtering optimization problems. Moreover, when optimizing the selected block of coordinates, forcing the other ones to remain constant leads to a new optimization problem that is often different from the original problem, and which can be significantly more difficult to solve.

In this paper, we introduce a generic distributed algorithm, referred to as the Distributed Adaptive Signal Fusion (DASF) algorithm, which can be used to solve generic linear DSFO problems. In each iteration of the algorithm, a node within the network is selected to receive compressed data from other nodes and to locally solve a lower-dimensional version of the original network-wide problem. A convenient property is that an instance of the same algorithm that solves the centralized network-wide optimization problem can be used to also solve the local (compressed) problems at each iteration. The compression allows to reduce bandwidth and energy usage in the network, while we still achieve convergence to the centralized solution. Moreover, since the locally constructed problem is of lower dimension, the computational cost required to solve it is also smaller compared to solving the network-wide problem. This makes the proposed algorithm also attractive in fully deterministic distributed settings where computational or memory resources are limited.

Various existing distributed algorithms can be shown to be special cases of our proposed DASF algorithm, including distributed algorithms for generalized eigenvalue decomposition (GEVD) \cite{bertrand2015dacgee}, spatial principal component analysis (PCA) \cite{bertrand2014dacmee}, least squares (LS), minimum mean square error (MMSE) or multi-channel Wiener filtering (MWF) \cite{doclo2009reduced,bertrand2010danse}, linearly constrained minimum variance (LCMV) beamforming \cite{markovich2012distributed,bertrand2013dlcmv,markovich2015optimal}, canonical correlation analysis (CCA) \cite{bertrand2015cca} and generalized CCA \cite{hovine2021distributed}. However, each of these algorithms was previously treated separately, with convergence proofs that were tailored to these specific cases. Our aim is to thoroughly define a unified algorithmic framework that contains these already existing algorithms but also extends to new DSFO problems. We also provide a toolbox available both in Matlab and Python to generate and validate new distributed algorithms within this framework \cite{musluoglu2022dsfotoolbox}.

The outline of this paper is as follows. In Section \ref{sec:prob_state}, we formally define the framework setting and the assumptions used throughout this paper. We then propose the DASF algorithm for fully-connected networks in Section \ref{sec:fc_dsf}, and later generalize it to general topologies in Section \ref{sec:ti_dsf}. Finally, we demonstrate the performance of the algorithm in a few new DSFO examples in Section \ref{sec:ex_sim}, thereby demonstrating the generalization properties of the DASF framework. We refer readers to the companion paper \cite{musluoglu2022unifiedp2} for detailed analyses and proofs.

\textbf{Notation:} Uppercase letters are used to represent matrices and sets, the latter in calligraphic script, while scalars, scalar-valued functions and vectors are represented by lowercase letters, the latter in bold. We use the notation $\chi_q^i$ to refer to a certain mathematical object $\chi$ (such as a matrix, set, etc.) at node $q$ and iteration $i$. The notation $\left(\chi^i\right)_{i\in\mathcal{I}}$ refers to a sequence of elements $\chi^i$ over every index $i$ in the ordered index set $\mathcal{I}$. If it is clear from the context (often in the case where $i$ is over all natural numbers), we omit the index set $\mathcal{I}$ and simply write $\left(\chi^i\right)_i$. A similar notation $\{\chi^i\}_{i\in\mathcal{I}}$ is used for non-ordered sets. Additionally, $I_Q$ denotes the $Q\times Q$ identity matrix, $\mathbb{E}[\cdot]$ the expectation operator, $\text{tr}(\cdot)$ the trace operator, $\textit{BlkDiag}\left(\cdot\right)$ the operator that creates a block-diagonal matrix from its arguments and $|\cdot|$ the cardinality of a set.

\section{Problem description}\label{sec:prob_state}
Consider a set of $K$ nodes, where $\mathcal{K}=\{1,\dots,K\}$ denotes the set of nodes. The nodes are interconnected in a connected graph where an edge between nodes $k$ and $q$ implies that these nodes can share data (e.g., via a wireless link). The set of neighbors of node $k$, i.e., the nodes that are connected to node $k$, is denoted by $\mathcal{N}_k$ (which excludes node $k$ itself).

Each node $k\in\mathcal{K}$ measures samples of a local $M_k$-channel signal $\mathbf{y}_k(t)\in\mathbb{R}^{M_k}$ at every time instance $t$. We define the network-wide $M$-channel signal $\mathbf{y}$ as 
\begin{equation}\label{eq:y_part}
  \mathbf{y}(t)=[\mathbf{y}_1^T(t),\dots,\mathbf{y}_K^T(t)]^T,
\end{equation}
with $M=\sum_k M_k$. All $\mathbf{y}_k$'s, and therefore also $\mathbf{y}$, are assumed to be (short-term) stationary and ergodic stochastic signals, such that their statistical properties can be properly estimated given a sufficiently large number of samples at different time instances.

The different channels of $\mathbf{y}$ can be spatially correlated across all nodes in the network, and we do not assume this correlation structure to be known. In a centralized setting, the channels of $\mathbf{y}$ can be linearly combined (fused) using a network-wide spatial filter $X\in\mathbb{R}^{M\times Q}$ with $Q$ output signals (with $Q\ll M$), where we typically aim to find an optimal $X$ such that the filter outputs $X^T\mathbf{y}\in\mathbb{R}^{Q}$ satisfy some optimality conditions. Typical examples of such filter design problems are listed in Table \ref{tab:ex_prob}, which can all be viewed as special cases of a general class of problems that will be formalized in the next subsection, which we refer to as (distributed) signal fusion optimization ((D)SFO) problems. This table is not exhaustive and various other signal fusion problems fit this framework (see e.g. \cite{cox1987robust,wouters2020multi}). Note that all of these examples require knowledge of the \textit{full} correlation matrix $R_{\mathbf{yy}}=\mathbb{E}[\mathbf{y}(t)\mathbf{y}^T(t)]$, which can only be estimated in a centralized setting where the data from all the nodes are collected in a single fusion center, allowing to estimate the correlation between any two channel pairs of $\mathbf{y}$. One of the key strengths of our proposed DASF framework is that it avoids such a data centralization (in the sense that there is never a node which has access to all the channels of $\mathbf{y}$, i.e., $R_{\mathbf{yy}}$ cannot be constructed), while still achieving the solution of the centralized problem.

\begin{table}
  \renewcommand{\arraystretch}{2.2}
  \caption{(D)SFO problems that are special cases of (\ref{eq:prob_g})}
  \label{tab:ex_prob}
  \begin{tablenotes}
    \item \textit{$\mathbf{y}$, $\mathbf{v}$ and $\mathbf{d}$ are multi-variate stochastic processes (signals). TRO is the trace ratio optimization problem and RR represents the ridge regression method. In the CCA case, the minimization is done with respect to $X$ and $W$.}
  \end{tablenotes}
  \begin{tabularx}{\columnwidth}{|s|b|z|}
  \hline
   Problem & Cost function to minimize & Constraints \\ \hhline{|=|=|=|}
   LCMV \cite{markovich2012distributed,bertrand2013dlcmv,markovich2015optimal} & $\mathbb{E}[||X^T\mathbf{y}(t)||^2]$ & $X^TB=H$ \\ \hline
   PCA \cite{bertrand2014dacmee} & $-\mathbb{E}[||X^T\mathbf{y}(t)||^2]$ & $X^TX=I_Q$ \\ \hline
   GEVD \cite{bertrand2015dacgee} & $-\mathbb{E}[||X^T\mathbf{y}(t)||^2]$ & $\mathbb{E}[X^T\mathbf{v}(t)\mathbf{v}^T(t)X]=I_Q$ \\ \hline
   TRO \cite{musluoglu2021distributed} & $-\frac{\mathbb{E}[||X^T\mathbf{y}(t)||^2]}{\mathbb{E}[||X^T\mathbf{v}(t)||^2]}$ & $X^TX=I_Q$ \\ \hline
   LS/MMSE /MWF \cite{doclo2009reduced,bertrand2010danse} & $\mathbb{E}[||\mathbf{d}(t)-X^T\mathbf{y}(t)||^2]$ & $X\in\mathbb{R}^{M\times Q}$ \\ \hline
   RR & $\mathbb{E}[||\mathbf{d}(t)-X^T\mathbf{y}(t)||^2]$ & $\text{tr}(X^TX)\leq \alpha^2$ \\ \hline
   CCA \cite{bertrand2015cca} & $-\mathbb{E}[\text{tr}(X^T\mathbf{y}(t)\mathbf{v}^T(t)W)]$ & \makecell{$\mathbb{E}[X^T\mathbf{y}(t)\mathbf{y}^T(t)X]=I_Q$\\$\mathbb{E}[W^T\mathbf{v}(t)\mathbf{v}^T(t)W]=I_Q$}  \\ \hline
  \end{tabularx}
\end{table}

Additionally, we consider inner products of the form $X^TB$, where $B\in\mathbb{R}^{M\times L}$ is a deterministic (i.e., fixed and time-independent) matrix. The LCMV example in Table \ref{tab:ex_prob} is an example where such an inner product with a deterministic matrix appears. Similar to $\mathbf{y}$ in (\ref{eq:y_part}), this term is defined as 
\begin{equation}\label{eq:B_part}
  B=[B_1^T,\dots,B_K^T]^T,
\end{equation}
where we only require that $B_k$ is known to node $k$. We note that the argument $B$ allows a deterministic representation of $\mathbf{y}$ in which multiple time samples of $\mathbf{y}$ are stored in the columns of $B$. Nevertheless, we make a distinction between both expressions to emphasize time-adaptive properties of the algorithm (see Section \ref{sec:practical_E}).

It is noted that some of the problems in Table \ref{tab:ex_prob} involve a second signal $\mathbf{v}:\mathbf{v}(t)=[\mathbf{v}_1^T(t),\dots,\mathbf{v}_K^T(t)]^T$ collected by the WSN (e.g., in the case of GEVD and CCA), and possibly another filter $W$ to be optimized (e.g., in the case of CCA). This additional signal could either be derived from the same set of sensors (e.g., a time-lagged version of $\mathbf{y}$ as in \cite{bertrand2015cca}, or observing $\mathbf{y}$ during two different regimes as in \cite{bertrand2015dacgee,wang2006common}), or they could come from different types of sensors with which the nodes are equipped.  In the remaining parts of this paper, we will typically consider the case of a single filter $X$, a single observed (multi-channel) sensor signal $\mathbf{y}$ and a single deterministic parameter $B$, yet all results can be easily generalized to multiple signals, parameters or filters. This generalization will be briefly addressed at the end of the next subsection and in Section \ref{sec:mult_var}.

It is important to note here that every other quantity of the problem that is not represented by an inner product with $X$ is assumed to be available at each node (e.g., $H$ and $\mathbf{d}$ in the LCMV and LS / MMSE / MWF examples in Table \ref{tab:ex_prob}). This means that each node is able to evaluate the objective and constraint functions of the optimization problem solved over the network if it has access to $X^T\mathbf{y}$ and $X^TB$.

\subsection{Scope of Signal Fusion Optimization Problems}\label{sec:opt_prob}
We first provide a generic description of the signal fusion optimization (SFO) problems that will be covered in this paper. While this description may seem rather exotic at first, we will provide several examples throughout the paper to illustrate how it contains many familiar problems as a special case. The SFO problems studied in this paper can be written in the following way:
\begin{equation}\label{eq:prob_g}
  \begin{aligned}
  \underset{X\in\mathbb{R}^{M\times Q}}{\text{minimize } } \quad & \varphi\left(X^T\mathbf{y}(t),X^TB\right)\\
  \textrm{subject to} \quad & \eta_j\left(X^T\mathbf{y}(t),X^TB\right)\leq 0,\;\textrm{ $\forall j\in\mathcal{J}_I$,}\\
    & \eta_j\left(X^T\mathbf{y}(t),X^TB\right)=0,\;\textrm{ $\forall j\in\mathcal{J}_E$,}
  \end{aligned}
\end{equation}
where $\varphi$ and the $\eta_j$'s are differentiable scalar- and real-valued functions, and the sets $\mathcal{J}_I$ and $\mathcal{J}_E$ represent the index sets of inequality and equality constraints respectively. Additionally, we define $\mathcal{J}=\mathcal{J}_I\cup\mathcal{J}_E$, and the number of constraints in total is given by  $J=|\mathcal{J}|$. 

An important observation is that $X$ always appears in an inner product with $\mathbf{y}$ or $B$, which corresponds to a signal fusion or spatial filtering operation. Furthermore, note that the functions that contain a stochastic signal $\mathbf{y}$ as an argument must contain an operator to translate this stochastic variable into a deterministic loss or constraint function (i.e., the functions in instances of Problem (\ref{eq:prob_g}) are deterministic, as any stochastic variable is converted into a deterministic value), for example through the use of an expectation operator (see Table \ref{tab:ex_prob}). In most practical cases, including those mentioned in Table \ref{tab:ex_prob}, the evaluation of $\varphi$ requires the knowledge or estimation of the network-wide spatial covariance matrix $R_{\mathbf{yy}}=E[\mathbf{y}(t)\mathbf{y}^T(t)]$. In this work, we assume that this matrix is unknown, in which case the spatial correlation between the nodes should be learned on the fly by the proposed distributed algorithm. 

The formulation (\ref{eq:prob_g}) covers a wide range of popular spatial filtering and signal processing problems, including those shown in Table \ref{tab:ex_prob}. For example, for the LCMV case (first example in Table \ref{tab:ex_prob}), we have $\varphi(X^T\mathbf{y}(t))=\mathbb{E}[||X^T\mathbf{y}(t)||^2]$, and $\eta_j(X^TB)=[X^TB-H]_j$ for each element $[X^TB-H]_j$ of the matrix $X^TB-H$. Note that quadratic terms of the form $X^TX$, which appear in some problems in Table \ref{tab:ex_prob}, should be seen as $(X^TB)\cdot(X^TB)^T$ with $B=I_M$. The reader is also referred to Section \ref{sec:ex_sim} in which a few extra examples are provided.

Finally, to simplify notation in various parts of this paper, we also define the differentiable functions $f$ and $h_j$'s, replacing $\varphi$ and $\eta_j$'s respectively, to describe Problem (\ref{eq:prob_g}) as a function of $X$ only, in which case $\mathbf{y}$ and $B$ should be viewed as internal function parameters:
\begin{equation}\label{eq:f_simplify}
  \begin{aligned}
    f(X)&\triangleq\varphi\left(X^T\mathbf{y}(t),X^TB\right),\\
    h_j(X)&\triangleq\eta_j\left(X^T\mathbf{y}(t),X^TB\right),\;\forall j\in\mathcal{J}.
  \end{aligned}
\end{equation}
Furthermore, we denote the constraint set of (\ref{eq:prob_g}) as $\mathcal{S}$, the complete solution set as $\mathcal{X}^*$ and a single solution as $X^*$, i.e., $X^*\in\mathcal{X}^*$.

\textbf{Note on further generalizations:} The problem description (\ref{eq:prob_g}) considers only one argument of each type $(X^T\mathbf{y}(t),X^TB)$ involving only one filter variable $X$, one stochastic signal $\mathbf{y}$ and one deterministic matrix $B$. However, this is merely for conciseness and intelligibility of the description of our framework, i.e., the framework can be straightforwardly generalized to multiple versions of variables and each of the two arguments in (\ref{eq:prob_g}), e.g., to also cover the cases of GEVD and CCA in Table \ref{tab:ex_prob}. Formally, the full scope of SFO problems we consider is
\begin{equation}\label{eq:prob_g_full}
  \begin{aligned}
  \underset{X^{(a)},\forall a}{\text{minimize } } & \varphi\left(X^{(a)T}\mathbf{y}^{(b)}(t),X^{(a)T}B^{(c)},...\right),\;\forall a,b,c\\
  \textrm{subject to } & \eta_j\left(X^{(a)T}\mathbf{y}^{(b)}(t),X^{(a)T}B^{(c)},...\right)\leq 0\;\textrm{ $\forall j\in\mathcal{J}_I$,}\\
    & \eta_j\left(X^{(a)T}\mathbf{y}^{(b)}(t),X^{(a)T}B^{(c)},...\right)=0\;\textrm{ $\forall j\in\mathcal{J}_E$.}
  \end{aligned}
\end{equation}
Additionally, even though we restrict ourselves to real-valued arguments, the results can be extended to complex-valued arguments (with real-valued cost functions) based on standard techniques such as those explained in \cite{hjorungnes2007complex,adali2014optimization,sayed2014adaptation}. Further extensions are possible where each node minimizes a different function, as in \cite{bertrand2010danse,bertrand2011lcdanse}, but this is beyond the scope of this paper.

\subsection{Adaptivity and Approximation in Practical Settings}\label{sec:practical_E}
The problems we are interested in typically involve an expectation operator over random signals with unknown distributions. In practical settings, the expectation operators in the objective and constraint functions of (\ref{eq:prob_g}) are usually approximated using sample averages \cite{van1988beamforming,blankertz2007optimizing,reed1974rapid,valaee1999localization,davila1994efficient}. The expectation operator over the distribution of $\mathbf{y}(t)$ for a generic deterministic function $g$ taking $\mathbf{y}(t)$ as an argument is then approximated\footnote{Convergence of the approximation to the true expectation for $N\rightarrow +\infty$ is studied in the stochastic optimization literature. We refer the reader to the sample average approximation method in particular \cite{kim2015guide}.} as
\begin{equation}\label{eq:approx_E}
  \mathbb{E}[g(\mathbf{y}(t))]\approx G(Y(t))\triangleq\frac{1}{N}\sum_{\tau=t}^{t+N-1}g(\mathbf{y}(\tau)),
\end{equation}
where $Y(t)=[\mathbf{y}(t),\dots,\mathbf{y}(t+N-1)]$ denotes a matrix that contains $N$ observations of $\mathbf{y}$, starting with the observation at time sample $t$. Here, it is assumed that the signal $\mathbf{y}$ is ergodic such that the expectation operators can be accurately approximated using (\ref{eq:approx_E}). In particular for the second-order statistics, which are commonly encountered in signal processing problems (see Table \ref{tab:ex_prob}), the approximation (\ref{eq:approx_E}) results in
\begin{equation}
  \mathbb{E}[\mathbf{y}(t)\mathbf{y}^T(t)]\approx\frac{1}{N}Y(t)Y^T(t).
\end{equation}
Table \ref{tab:ex_prob_E} gives the practical approximations of some commonly encountered functions in problems of interest (including expressions found in Table \ref{tab:ex_prob}) using (\ref{eq:approx_E}). Note that the stationarity assumption removes any time dependence from the problems of interest such that any window of contiguous time samples of $\mathbf{y}$ can be used. 

Throughout this paper, we assume stationarity for mathematical tractability, which is a common assumption in the analysis of adaptive filters \cite{van1988beamforming,diniz2019adaptive,sayed2011adaptive,sayed2014adaptive,chen2003asymptotic}. However, in practice, we do not require the signals to be fully stationary, as long as the dynamics in the underlying signal statistics are sufficiently slow, such that (\ref{eq:approx_E}) gives a reasonable approximation. In this case, the solution $X^*$ of (\ref{eq:prob_g}) becomes time-dependent. When the DASF framework is applied in such an adaptive context, the targeted instance of Problem (\ref{eq:prob_g}) is effectively replaced by its sample average counterpart, where the statistics are regularly re-estimated in a block-based fashion, i.e., each time a new block of $N$ samples becomes available. As we will explain in Sections \ref{sec:fc_dsf} and \ref{sec:ti_dsf}, every new block of $N$ samples will initiate one new iteration of the DASF algorithm, i.e., the iterations of DASF are spread over different sample blocks, such that the algorithm becomes time-recursive (implicitly assuming $G(Y(t))\approx G(Y(t+N))$). We will demonstrate in Section \ref{sec:adaptivity} through an example that the proposed DASF algorithm can indeed track slow changes in the signal statistics, and is able to recover from abrupt changes.

\begin{table}
  \renewcommand{\arraystretch}{2}
  \caption{Practical Approximations of Common Functions}
  \label{tab:ex_prob_E}
  \begin{tablenotes}
    \item \textit{Approximations computed in practice to evaluate some commonly encountered functions using (\ref{eq:approx_E}) and $N$ observations of the stochastic signals, e.g., $\{\mathbf{y}(\tau)\}_{\tau=t}^{t+N-1}$. $Y(t)$ denotes the sample matrix of $\mathbf{y}$ for $N$ observations, $Y=[\mathbf{y}(t),\dots,\mathbf{y}(t+N-1)]$, while $V$ and $D$ are similarly defined for $\mathbf{v}$ and $\mathbf{d}$ respectively.}
  \end{tablenotes}
  \begin{tabularx}{\columnwidth}{|y|x|}
  \hline
   $\mathbb{E}[g(\mathbf{y}(t))]$ & $G(Y(t))$ \\ \hhline{|=|=|}
   $\mathbb{E}[||X^T\mathbf{y}(t)||^2]$ & $||X^TY(t)||_F^2/N$\\ \hline
   $\mathbb{E}[X^T\mathbf{y}(t)\mathbf{y}^T(t)X]$ & $X^TY(t)Y^T(t)X/N$ \\ \hline
   $\mathbb{E}[||\mathbf{d}(t)-X^T\mathbf{y}(t)||^2]$ & $||D(t)-X^TY(t)||^2_F/N$ \\ \hline
   $\mathbb{E}[\text{tr}(X^T\mathbf{y}(t)\mathbf{v}^T(t)W)]$ & $\text{tr}(X^TY(t)V^T(t)W)/N$  \\ \hline
  \end{tabularx}
\end{table}

\subsection{General Assumptions}\label{sec:asmp}
In order for our DASF algorithm to be applicable and achieve convergence and optimality, the problem (\ref{eq:prob_g}) must satisfy some sufficient conditions, which are listed in this subsection for completeness, while the technical details are explained in the companion paper \cite{musluoglu2022unifiedp2}. It is noted that these conditions are usually satisfied for all examples listed in Table \ref{tab:ex_prob} (except in some contrived cases) and we refer the reader to \cite{musluoglu2022unifiedp2} for some examples on how these conditions can be checked in practice. In addition to these sufficient conditions, we restate the implicit assumption from Section \ref{sec:opt_prob} that the functions $f$ and $h_j$ in (\ref{eq:f_simplify}) are smooth functions, i.e., they are continuously differentiable over the variable $X$ on their respective domain, or equivalently, the functions $\varphi$ and $\eta_j$ are continuously differentiable over $X$ for any $\mathbf{y}$ and $B$.

\begin{asmp}\label{asmp:well_posed}
  The targeted instance of Problem (\ref{eq:prob_g}) is well-posed\footnote{The notion of (generalized Hadamard) well-posedness we require is based on \cite{hadamard1902problemes,zhou2005hadamard}. The main difference is that we require the map from the parameter (inputs of the problem) space to the solution space to be continuous instead of upper semicontinuous, which is required for the convergence proof of the DASF algorithm. These technical details are presented in \cite{musluoglu2022unifiedp2}.}, in the sense that the solution set is not empty and varies continuously with a change in the parameters of the problem.
\end{asmp}

\noindent Note that since in practice, the DASF algorithm will be used for solving a particular instance of Problem (\ref{eq:prob_g}), Assumption \ref{asmp:well_posed} is only required for that particular problem and not all problems within the scope of the framework. As an example, consider the PCA problem of Table \ref{tab:ex_prob}. It can be shown that the PCA problem satisfies Assumption \ref{asmp:well_posed} if the covariance matrix of $\mathbf{y}$ is positive definite and its $Q+1$ largest eigenvalues are all distinct.

\begin{asmp}\label{asmp:kkt}
  The linear independence constraint qualifications (LICQ) hold at the solutions of Problem (\ref{eq:prob_g}), i.e., the solutions satisfy the Karush-Kuhn-Tucker (KKT) conditions.
\end{asmp}
\noindent If $X^*$ is a solution of Problem (\ref{eq:prob_g}), then Assumption \ref{asmp:kkt} implies that the gradients $\nabla_X h_j(X^*)$, $j\in\mathcal{J}^*$, are linearly independent\footnote{A set of matrices $\{A_j\}_{j\in\mathcal{J}}$ is linearly independent when $\sum_{j\in\mathcal{J}}\alpha_jA_j=0$ is satisfied if and only if $\alpha_j=0$, $\forall j\in\mathcal{J}$, or equivalently, when $\{\text{vec}(A_j)\}_{j\in\mathcal{J}}$ is a set of linearly independent vectors, where $\text{vec}(\cdot)$ is the vectorization operator.}, where $\mathcal{J}^*\subseteq \mathcal{J}$ is the set of all indices $j$ for which $h_j(X^*)=0$. We refer the reader to \cite{peterson1973review} for further details on constraint qualifications. If there is no constraint function $\eta_j$ in Problem (\ref{eq:prob_g}), Assumption \ref{asmp:kkt} implies that $\nabla_X f(X^*)=0$.
\begin{asmp}
  $f$ has compact sublevel sets in $\mathcal{S}$, i.e., for all $m\in\mathbb{R}$, $\{X\in\mathcal{S}\;|\;f(X)\leq m\}$ is compact.
\end{asmp}
\noindent Note that in $\mathbb{R}^{M\times Q}$, compactness is equivalent to closedness and boundedness of a set, which is a relatively mild condition. In fact, as is shown in \cite{musluoglu2022unifiedp2}, the assumption can be further relaxed by only requiring that at least one sublevel set $\{X\in\mathcal{S}\;|\;f(X)\leq f(X^0)\}$ is compact, where $X^0$ is the initialization point of the DASF algorithm. 

Moreover, the convergence proof in \cite{musluoglu2022unifiedp2} requires an additional sufficient condition that is akin to the LICQ. We postpone its definition to \cite{musluoglu2022unifiedp2} because of its technical nature, and because it is a relatively mild condition, which is generally satisfied in practice. Nevertheless, an important implication of this condition, which is relevant to disclose at this point, will be that it imposes an upper bound on the number of constraints $J$ the problem is allowed to have (see \cite{musluoglu2022unifiedp2}):
\begin{equation}\label{eq:upper_bound}
  J\leq \min\left(\frac{Q^2}{K-1}\sum_{k\in\mathcal{K}}|\mathcal{N}_k|,\;(1+\min_{k\in\mathcal{K}}|\mathcal{N}_k|)Q^2\right).
\end{equation}
Here, $K$ is the total number of nodes, $|\mathcal{N}_k|$ is the number of neighbors of node $k$, and $Q$ is the number of columns of $X$.

We also make the implicit assumption that a (centralized) algorithm is available to solve (with arbitrary accuracy) the targeted problem instance, i.e., the instance of (\ref{eq:prob_g}) for which we aim to design a distributed algorithm. This is a reasonable premise, since it makes little sense to design a distributed algorithm for a problem that cannot even be solved in a centralized setting. Nevertheless, it is important as our DASF framework will use the same solver to find solutions of compressed versions of the targeted problem at each node, as will be explained next. We note that there are no restrictions on the solver, which can be chosen freely (e.g., closed-form solutions, steepest descent methods, interior point methods, trust region methods, etc.).

\section{Distributed Adaptive Signal Fusion in Fully-Connected Networks (FC-DASF)}\label{sec:fc_dsf}

For the sake of an easier exposition, let us first consider the special case of a fully-connected network where data transmitted by any node is received by every other node (more general topologies are treated in Section \ref{sec:ti_dsf}). As the optimization problem (\ref{eq:prob_g}) depends on the full signal $\mathbf{y}$ and its (unknown) spatial correlation across the nodes, the nodes would need to share some information between each other. However, sharing the full signal $\mathbf{y}$ would require significant bandwidth and consume a large amount of power. Instead, each node will linearly compress its local $M_k-$channel signal into an $R-$channel signal (with $R<M_k$) before broadcasting it to the other nodes in the network. At iteration $i$, the compression is done by multiplying the signal $\mathbf{y}_k$ at node $k$ with an $M_k\times R$ matrix $F_k^i$, which we refer to as the local compressor at node $k$. The $R-$dimensional compressed signal resulting from this compressive spatial filtering can be written as
\begin{equation}\label{eq:compress}
  \widehat{\mathbf{y}}_k^i(t)\triangleq F_k^{iT}\mathbf{y}_k(t).
\end{equation}
The deterministic parameters $B_k$ are similarly compressed to obtain $\widehat{B}_k^i$, which are also broadcast between nodes.
The index $i$ emphasizes that these compressors are not constant and will be iteratively updated across time. Note that the iteration index $i$ is also added to the stochastic signal $\widehat{\mathbf{y}}_k$ in order to indicate that the content of the signal (and hence its statistics) changes in each iteration due to an update of the underlying compression matrix. However, this does not imply that we iterate over a single batch of samples of this signal, i.e., an update from $\widehat{\mathbf{y}}_k^i$ to $\widehat{\mathbf{y}}_k^{i+1}$ (or $F_k^i$ to $F_k^{i+1}$) only affects future samples of $\widehat{\mathbf{y}}_k$ that are collected after performing the update. As a result, the compressor $F_k$ operates as a block-adaptive filter, which updates its coefficients after every new block of $N$ samples, such that $F_k^i$ operates on the samples in the data matrix $Y(iN)$, whereas $F_k^{i+1}$ operates on the samples in $Y(iN+N)$, where $Y(t)$ is defined as in (\ref{eq:approx_E}). In other words, an update of the compressor $F_k$ affects how future samples of $\mathbf{y}_k$ are compressed, yet previously collected samples will not be re-compressed or re-broadcast.

Equation (\ref{eq:compress}) results in a compression ratio of $M_k/R$. The compression implies that the nodes do not have full access to the network-wide signal $\mathbf{y}$. Nevertheless, we will show that an optimal solution $X^*\in\mathcal{X}^*$ of Problem (\ref{eq:prob_g}) can be achieved if $R=Q$, where $Q$ is the number of columns of $X$, i.e., the number of output signals of the filter $X$ (in many cases only a single-channel output is desired, such that $R=Q=1$). In a fully-connected network, this implies that we must assume that $Q<M_k$ in order to achieve a bandwidth reduction at node $k$. However, in the case of more general topologies, we can also achieve this for $Q\geq M_k$ (see Section \ref{sec:ti_dsf}).

After introducing the DASF algorithm for (\ref{eq:prob_g}) in the next subsection, we will provide insights on the relationship between the network-wide problem and the problems solved at each node in Section \ref{sec:central_local}. Extensions to the more general form (\ref{eq:prob_g_full}) will be presented in Section \ref{sec:mult_var}.

\subsection{Algorithm Derivation}\label{sec:algo_desc}
Consider the partitioning of the optimization variable $X$ in per-node sub-blocks, i.e.,
\begin{equation}\label{eq:X_part}
  X=[X_1^T,\dots,X_K^T]^T,
\end{equation}
where $X_k\in\mathbb{R}^{M_k\times Q}$. This way, every $X_k$ has a corresponding local signal $\mathbf{y}_k$, i.e., the part of $X$ that is applied to $\mathbf{y}_k$ in the expression $X^T\mathbf{y}$, such that we can write
\begin{equation}
  X^T\mathbf{y}(t)=\sum_{k\in\mathcal{K}}X_k^T\mathbf{y}_k(t).
\end{equation}
A similar observation for the local deterministic terms $B_k$ implies that we can write the objective $\varphi$ using the local filters and data:
\begin{align}
  f(X)&=\varphi(X^T\mathbf{y}(t),X^TB)\nonumber \\
  &=\varphi\left(\sum_{k\in\mathcal{K}}X_k^T\mathbf{y}_k(t),\sum_{k\in\mathcal{K}}X_k^TB_k\right).
\end{align}
The constraint functions $\eta_j$ can also be written in a similar way. Therefore, we are able to express the full optimization problem (\ref{eq:prob_g}) using $X_k$'s, $\mathbf{y}_k$'s and $B_k$'s. The main idea behind the DASF framework is to partially reconstruct and solve a compressed version of Problem (\ref{eq:prob_g}) at any selected node, which is referred to as the ``updating node'', and which changes from iteration to iteration. Let us now set $R=Q$ and
\begin{equation}\label{eq:compressor}
  F_k^i=X_k^i,
\end{equation}
where $X_k^i$ corresponds to the local estimate of $X_k$ at node $k$ at iteration $i$. Stacking them together as in (\ref{eq:X_part}), we obtain the estimate $X^i$ of the global variable $X$ at iteration $i$. This means that, within the DASF algorithm, the $X_k^i$'s act both as compressors and as part of the optimization variable. Combining (\ref{eq:compress}) and (\ref{eq:compressor}), the compressed signal that is broadcast by node $k$ can be written as
\begin{equation}\label{eq:y_hat}
  \widehat{\mathbf{y}}^i_k(t)\triangleq X_k^{iT}\mathbf{y}_k(t)\in\mathbb{R}^{Q},
\end{equation}
which implies that the network-wide filter output at iteration $i$ can be computed as the sum of the compressed signals in (\ref{eq:y_hat}), i.e.,
\begin{equation}\label{eq:y_Xy}
  \widehat{\mathbf{y}}^i(t)\triangleq X^{iT}\mathbf{y}(t)=\sum_{k\in\mathcal{K}}X_k^{iT}\mathbf{y}_k(t)=\sum_{k\in\mathcal{K}}\widehat{\mathbf{y}}_k^i(t).
\end{equation}
In a similar way, the compressed deterministic terms at node $k$ are
\begin{equation}\label{eq:B_hat}
  \widehat{B}^i_k\triangleq X_k^{iT}B_k\in\mathbb{R}^{Q\times L},
\end{equation}
and we have
\begin{equation}\label{eq:B_XB}
  \widehat{B}^i\triangleq X^{iT}B=\sum_{k\in\mathcal{K}}X_k^{iT}B_k=\sum_{k\in\mathcal{K}}\widehat{B}_k^i.
\end{equation}
Since the network is assumed to be fully-connected, each node has access to (observations of) all signals $\widehat{\mathbf{y}}_k^i$, such that each node can compute the filter outputs (\ref{eq:y_Xy}). Let $q$ be the updating node at iteration $i$, and define the stacked vector containing all available signals at node $q$ as
\begin{equation}\label{eq:y_tilde}
  \widetilde{\mathbf{y}}_q^i(t)\triangleq[\mathbf{y}_q^T(t),\widehat{\mathbf{y}}_1^{iT}(t),\dots,\widehat{\mathbf{y}}_{q-1}^{iT}(t),\widehat{\mathbf{y}}_{q+1}^{iT}(t),\dots,\widehat{\mathbf{y}}_K^{iT}(t)]^T,
\end{equation}
which contains $\widetilde{M}_q\triangleq M_q+Q(K-1)$ channels. Note that node $q$ only has access to uncompressed observations of $\mathbf{y}_q$ and a corresponding batch of compressed observations of all the other nodes. Similarly, the matrix containing all available deterministic terms is obtained by stacking $B_q$ which is available at node $q$ and the compressed $\widehat{B}_k^i$'s received from other nodes:
\begin{equation}\label{eq:B_tilde}
  \widetilde{B}^i_q\triangleq [B_q^T,\widehat{B}_1^{iT},\dots,\widehat{B}_{q-1}^{iT},\widehat{B}_{q+1}^{iT},\dots,\widehat{B}_K^{iT}]^T,
\end{equation}
which is an $\widetilde{M}_q\times L$ matrix.
Based on (\ref{eq:y_tilde}) and (\ref{eq:B_tilde}), we define a new local variable $\widetilde{X}_q\in\mathbb{R}^{\widetilde{M}_q\times Q}$ at node $q$, such that we are able to formulate a local optimization problem using only data available at node $q$ at iteration $i$:
\begin{equation}\label{eq:loc_prob_g}
  \begin{aligned}
    \underset{\widetilde{X}_q\in\mathbb{R}^{\widetilde{M}_q\times Q}}{\text{minimize } } \quad & \varphi(\widetilde{X}_q^T\widetilde{\mathbf{y}}_q^i(t),\widetilde{X}_q^T\widetilde{B}_q^i)\\
  \textrm{subject to} \quad & \eta_j(\widetilde{X}_q^T\widetilde{\mathbf{y}}_q^i(t),\widetilde{X}_q^T\widetilde{B}_q^i)\leq 0\;\textrm{ $\forall j\in\mathcal{J}_I$},\\
   & \eta_j(\widetilde{X}_q^T\widetilde{\mathbf{y}}_q^i(t),\widetilde{X}_q^T\widetilde{B}_q^i)=0\;\textrm{ $\forall j\in\mathcal{J}_E$}.
  \end{aligned}
\end{equation}
A key observation here is the similarity between (\ref{eq:loc_prob_g}) and (\ref{eq:prob_g}). This means that node $q$ can locally apply the same solver as the one used for solving the centralized problem, albeit on a problem of smaller size. Note that this implies that the computational cost required to solve Problem (\ref{eq:loc_prob_g}) is smaller compared to solving (\ref{eq:prob_g}). 

At iteration $i$, node $q$ solves the local problem (\ref{eq:loc_prob_g}), and we denote its solution as $\widetilde{X}_q^{i+1}$, which can be partitioned as
\begin{align}\label{eq:X_tilde_i1}
  \widetilde{X}_q^{i+1}&=\\ \nonumber
  &[X_q^{(i+1)T},G_1^{(i+1)T},\dots,G_{q-1}^{(i+1)T},G_{q+1}^{(i+1)T},\dots,G_K^{(i+1)T}],
\end{align}
where $X_q^{i+1}$ is $M_q\times Q$ and each $G_k^{i+1}$ is $Q\times Q$. By comparing (\ref{eq:X_tilde_i1}) with (\ref{eq:y_tilde}), we see that $G_k^{i+1}$ refers to the part of $\widetilde{X}_q$ that is multiplied with the received compressed data $\widehat{\mathbf{y}}_k^{i}$ from node $k$ in the inner product $\widetilde{X}_q^{T}\widetilde{\mathbf{y}}_q^{i}(t)$ in (\ref{eq:loc_prob_g}). Since $\widehat{\mathbf{y}}_k^i(t)=X_k^{iT}\mathbf{y}_k(t)$, we can instead multiply the compressor $X_k^i$ at node $k$ with this matrix $G_k^{i+1}$. As a result, each node $k$ in the network updates its local $X_k$ as
\begin{equation}\label{eq:full_upd}
  X_k^{i+1}=\begin{cases}X_q^{i+1} & \text{if $k=q$}\\
  X_k^{i}G_k^{i+1} & \text{if $k\neq q$,}\end{cases}
\end{equation}
where $X_q^{i+1}$ and $G_k^{i+1}$ are obtained from the partitioning (\ref{eq:X_tilde_i1}) of $\widetilde{X}_q^{i+1}$. Since the updating node $q$ does not have access to the $X_k^i$ of the other nodes $k\neq q$, it needs to communicate the matrices $G_k^{i+1}$ to the other nodes, so that they can update their local $X_k$ as well\footnote{Note that the communication cost to transmit these $G_k$ matrices is negligible compared to the transmission of a batch of observations of $\widehat{\mathbf{y}}_k$'s (see also Remark \ref{rem:batch}).}. A block diagram of this process is provided in Figure \ref{fig:block_diagram}.

If the minimization (\ref{eq:loc_prob_g}) has multiple solutions, an ambiguity exists on the choice of the local variable, which can be resolved by selecting a specific solution at each iteration. We propose to select the solution $\widetilde{X}_q^{i+1}$ for which the distance $||\widetilde{X}_q^{i+1}-\widetilde{X}_q^i||_F$ is minimal, where $\widetilde{X}_q^i$ is defined as
\begin{equation}\label{eq:X_fixed}
  \widetilde{X}_q^{i}=[X_q^{iT},I_Q,\dots,I_Q]^T.
\end{equation}
The choice of the Frobenius norm $||\cdot||_F$ as a distance metric is arbitrary. Other distance functions $d$ can also be used (and might be better suited for the specific instance of Problem (\ref{eq:prob_g}) at hand), as long as they are continuous and satisfy $d(X,Y)=0\iff X=Y$. These conditions on $d$ are needed for the convergence of the proposed method, as explained in \cite{musluoglu2022unifiedp2}.

\begin{figure}[t]
  \includegraphics[width=0.48\textwidth]{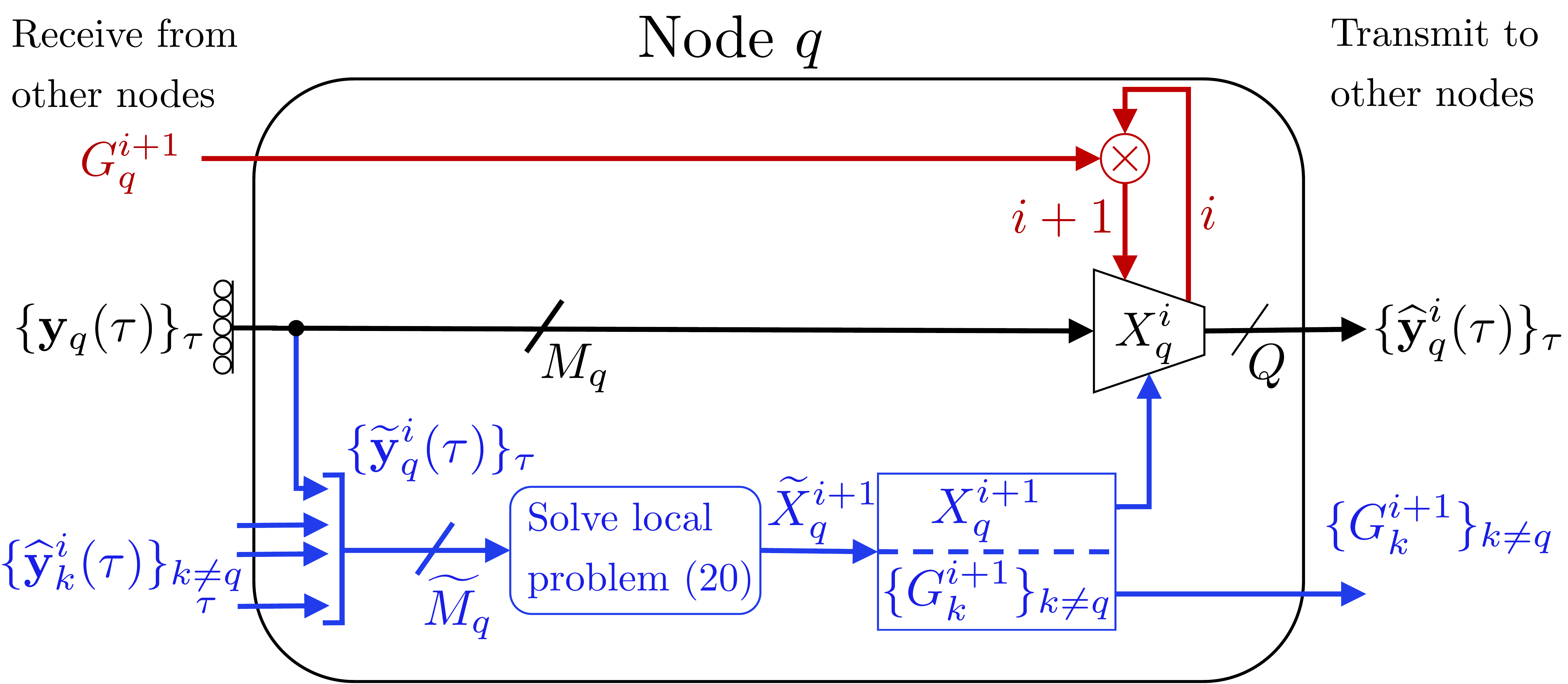}
  \caption{Block diagram representation of the steps followed by a given node $q$ in the FC-DASF algorithm. The black part is executed for any sample time $t$ at each node. The red and blue parts are only executed at each iteration increment $i\rightarrow i+1$, where the blocks in blue are only executed when node $q$ is the updating node. Otherwise, the part in red is carried out. Node $q$ has always access to its own signal samples $\mathbf{y}_q(t)$ measured at its own sensors (represented by rings, in black), while the compressed signal samples $\widehat{\mathbf{y}}_k^i(t)$ are transmitted to node $q$ by the respective nodes $k$ (represented by arrows in blue). For intelligibility, we omitted the data flow of the expression $X^TB$ from the diagram.}
  \label{fig:block_diagram}
\end{figure}

\begin{figure}[!t]
  \removelatexerror
  \DontPrintSemicolon
  \begin{algorithm}[H]
  \caption{Fully-Connected Distributed Adaptive Signal Fusion (FC-DASF) Algorithm\;
  Code available in \cite{musluoglu2022dsfotoolbox}}\label{alg:fc_dsfo}
  \SetKwInOut{Output}{output}
  \Output{$X^*$}
  \BlankLine
  Initialize $X^0$, $i\gets0$.\;
  \Repeat
  {
  Choose the updating node as $q\gets (i\mod K)+1$.\;
  1) Every node $k$ collects a new batch of $N$ samples of $\mathbf{y}_k$ (see Remark \ref{rem:batch}), compresses these to $N$ samples of $\widehat{\mathbf{y}}^{i}_k$ using (\ref{eq:y_hat}) and transmits them to node $q$.\;
  $\widehat{B}_k^i$ is computed using (\ref{eq:B_hat}) and transmitted to node $q$.\;

  \At{Node $q$}
  {
    2a) Compute $\widetilde{X}_q^{i+1}$ as the solution of (\ref{eq:loc_prob_g}). If the solution is not unique, select the solution which minimizes $||\widetilde{X}_q^{i+1}-\widetilde{X}_q^i||_F$ with $\widetilde{X}_q^i$ defined in (\ref{eq:X_fixed}).\;
    2b) Partition $\widetilde{X}^{i+1}_q$ as in (\ref{eq:X_tilde_i1}).\;
    2c) Transmit $G_k^{i+1}$ to node $k$ for every $k\neq q$.\;
  }
  
  3) Every node updates $X_k^{i+1}$ according to (\ref{eq:full_upd}).\;
  
  $i\gets i+1$\;
  }
  \end{algorithm}
  \small{\textbf{Note:} Each iteration uses a different batch of $N$ samples in step 1, i.e., the iterations can be spread over different time segments in order to avoid retransmitting the same batch of $N$ samples across the network (see also Remark \ref{rem:batch}). This makes the sample time index $t$ coupled to the iteration index $i$.}
  \hrule
\end{figure}

All the steps of the FC-DASF algorithm as explained above are summarized in Algorithm \ref{alg:fc_dsfo}. Algorithm \ref{alg:fc_dsfo} converges under mild technical conditions to an optimal filter $X^*$ solving the network-wide problem (\ref{eq:prob_g}). The convergence results with detailed analysis and proofs can be found in the companion paper \cite{musluoglu2022unifiedp2}.

\begin{rem}\label{rem:batch}
  It is noted that a transmission of the compressed signal $\widehat{\mathbf{y}}^i_k$ at node $k$ corresponds in practice to transmitting a batch with the $N$ most recent time samples of $\widehat{\mathbf{y}}^i_k$. This allows for the receiving node $q$ to estimate the necessary signal statistics to evaluate or optimize (\ref{eq:loc_prob_g}), where a larger value of $N$ results in a closer approximation of the true value (remember that the objective function in (\ref{eq:prob_g}) and (\ref{eq:loc_prob_g}) has a built-in operator to transform the stochastic signal $\mathbf{y}$ into a deterministic loss, which in practice is usually replaced with an average over $N$ samples, as in Table \ref{tab:ex_prob_E}). Leveraging the (short-term) stationarity assumption, different batches of $N$ samples are used in each iteration, such that the communication bandwidth becomes independent of the number of iterations. Therefore, Algorithm \ref{alg:fc_dsfo} behaves similarly to an adaptive filter which learns over time how to optimally filter newly observed samples (based on past samples). In a tracking context where the signal statistics of $\mathbf{y}$ change over time, the algorithm can still be applied if the statistics change slower than the convergence speed of the algorithm. 
  
  The DASF framework could in principle also be applied in a batch-mode (non-adaptive) framework, in which all operations are performed entirely on a single batch of samples (instead of spreading out the iterations over different sample batches of length $N$). In this case, the argument $X^T\mathbf{y}$ can be dropped and the argument $X^TB$ can be used to represent the batch of samples, in which all available samples of $\mathbf{y}$ are stored in the columns of $B$.
\end{rem}

\begin{rem}
  Although each node is able to communicate with every other node in a fully-connected network, Algorithm \ref{alg:fc_dsfo} is still distributed in nature. Indeed, the network-wide data is never centralized, i.e., the updating nodes $q$ have never access to $\mathbf{y}$ or $B$, but only to $\widetilde{\mathbf{y}}_q^i$ and $\widetilde{B}_q^i$, and therefore cannot estimate any network-wide statistics such as $\mathbb{E}[\mathbf{y}(t)\mathbf{y}^T(t)]$, whereas a centralized solver has access to this information.
\end{rem}

\subsection{Link between the Central and Local Problems}\label{sec:central_local}
In this subsection, we further explain the link between the central problem (\ref{eq:prob_g}) and the local problems (\ref{eq:loc_prob_g}) at the updating node $q$, where the latter can be viewed as a parameterized version of the former. This will provide some additional insights and show some useful properties of the DASF framework. Their relationship can be described by means of the transformation matrix:
\begin{equation}\label{eq:cqi}
  C_q^{i}=
\left[
{\renewcommand{\arraystretch}{1.8}%
\begin{array}{c|c|c}
0 & \Theta_{<q}^{i} & 0\\
\cline{1-3}
I_{M_q} & 0 & 0\\
\cline{1-3}
0 & 0 & \Theta_{>q}^{i}
\end{array}}
\right]\in\mathbb{R}^{M\times \widetilde{M}_q},
\end{equation}
where $\Theta_{<q}^i=\textit{BlkDiag}(X_1^i,\dots,X_{q-1}^{i})$ and $\Theta_{>q}^i=\textit{BlkDiag}(X_{q+1}^i,\dots,X_{q+1}^i)$. Then, from (\ref{eq:y_hat}) and (\ref{eq:y_tilde}), one can validate that
\begin{equation}\label{eq:compress_y}
  \widetilde{\mathbf{y}}_q^i(t)=C_q^{iT}\mathbf{y}(t),
\end{equation}
while (\ref{eq:B_hat}) and (\ref{eq:B_tilde}) result in
\begin{equation}\label{eq:compress_B}
  \widetilde{B}_q^i=C_q^{iT}B.
\end{equation}
Using these relationships in (\ref{eq:loc_prob_g}), we see that
\begin{align}\label{eq:equivalence}
  \varphi(\widetilde{X}_q^T\widetilde{\mathbf{y}}_q^i(t),\widetilde{X}_q^T\widetilde{B}_q^i)&=\varphi\left(\widetilde{X}_q^T(C_q^{iT}\mathbf{y}(t)),\widetilde{X}_q^T(C_q^{iT}B)\right)\nonumber \\
  &=\varphi\left((C_q^i\widetilde{X}_q)^T\mathbf{y}(t),(C_q^i\widetilde{X}_q)^TB\right)\nonumber \\
  &=f(C_q^i\widetilde{X}_q).
\end{align}
Similarly, we find that, $\forall j\in\mathcal{J}$:
\begin{align}\label{eq:equivalence_h}
  \eta_j(\widetilde{X}_q^T\widetilde{\mathbf{y}}_q^i(t),\widetilde{X}_q^T\widetilde{B}_q^i)&=\eta_j\left((C_q^i\widetilde{X}_q)^T\mathbf{y}(t),(C_q^i\widetilde{X}_q)^TB\right)\\ \nonumber
  &=h_j(C_q^i\widetilde{X}_q).
\end{align}
This implies that the local optimization variable $\widetilde{X}_q$ defines a parameterization of the global variable $X$. Indeed, if we define
\begin{equation}\label{eq:X_tilde_part}
  \widetilde{X}_q=[X_q^T,G_1^T,\dots,G_{q-1}^T,G_{q+1}^T,\dots,G_K^T]^T,
\end{equation}
where $X_q$ is $M_q\times Q$ and every $G_k$ is $Q\times Q$, we have at iteration $i$ (for the updating node $q$)

\begin{equation}\label{eq:X_param}
  X=C_q^i\widetilde{X}_q=\begin{bmatrix}
    X_1^i\boxed{G_1}\\\vdots\\X_{q-1}^i\boxed{G_{q-1}}\\\boxed{X_q}\\X_{q+1}^i\boxed{G_{q+1}}\\\vdots\\X_K^i\boxed{G_K}
\end{bmatrix}.
\end{equation}
Note that only the framed variables in (\ref{eq:X_param}) appear as optimization variables in the local problem (\ref{eq:loc_prob_g}), which is clear from comparing (\ref{eq:X_tilde_i1}) with (\ref{eq:X_tilde_part}). This shows that the updating node $q$ only has the full freedom to update $X_q$, i.e., its local compressor. The remaining parts of $X$ can only change up to a multiplication from the right by a matrix $G_k$, $k\neq q$ when it is node $q$'s turn to solve the local problem. Therefore, by sequentially changing the updating node across iterations, we allow every node to fully update its own local compressor while only manipulating the other sub-blocks of $X$ within their respective column spaces. 

The solution $\widetilde{X}_q^{i+1}$ of the local problem (\ref{eq:loc_prob_g}) at node $q$ and iteration $i$, which can also be written as
\begin{align}\label{eq:comp_X_tilde_parametric}
  \widetilde{X}_q^{i+1}\triangleq\;&\underset{\widetilde{X}_q\in\widetilde{\mathcal{S}}_q^i}{\text{argmin }}f\left(C_q^i\widetilde{X}_q\right),\nonumber \\
  =\;&\underset{\widetilde{X}_q\in\widetilde{\mathcal{S}}_q^i}{\text{argmin }}\varphi\left(\widetilde{X}_q^T\widetilde{\mathbf{y}}_q^i(t),\widetilde{X}_q^T\widetilde{B}_q^i\right),
\end{align}
where $\widetilde{\mathcal{S}}_q^i$ denotes the constraint set of (\ref{eq:loc_prob_g}), defines a new point $X^{i+1}$ for the global problem (\ref{eq:prob_g}) via (\ref{eq:X_param}). 

In the following lemma, we show that the global variable $X^i$ produced by the DASF algorithm\footnote{The proof of Lemma \ref{lem:X_in_constraints} also holds for the general topology-independent DASF algorithm in Section \ref{sec:ti_dsf}.} always satisfies the global constraint set $\mathcal{S}$ for any iteration $i>0$.
\begin{lem}\label{lem:X_in_constraints}
  For any iteration $i>0$,
  \begin{equation}\label{eq:loc_glob}
    \widetilde{X}_q\in\widetilde{\mathcal{S}}_q^i\iff C_q^i\widetilde{X}_q\in\mathcal{S}.
  \end{equation}
  In particular, $X^i\in\mathcal{S}$ and $\widetilde{X}_q^{i}\in\widetilde{\mathcal{S}}_q^i$ for all $i>0$.
\end{lem}
\begin{proof}
  From (\ref{eq:equivalence_h}), it follows automatically that any point $\widetilde{X}_q$ in the constraint set of the local problem (\ref{eq:loc_prob_g}) has a corresponding point $X=C_q^{i}\widetilde{X}_q$ in the constraint set of the global problem (\ref{eq:prob_g}). This implies that, if $\widetilde{X}_q$ is a feasible point of (\ref{eq:loc_prob_g}), the point $X$ parameterized by $\widetilde{X}_q$, such that $X=C_q^i\widetilde{X}_q$, is a feasible point of (\ref{eq:prob_g}), and vice versa, which proves (\ref{eq:loc_glob}).

  From (\ref{eq:X_tilde_part})-(\ref{eq:X_param}), we find that the point $X^i$ (before the update at iteration $i$) is equal to $X^{i}=C_q^i\widetilde{X}_q^i$ with $\widetilde{X}_q^i$ defined in (\ref{eq:X_fixed}). Similarly, we know (by construction) that $X^{i+1}=C_q^i\widetilde{X}_q^{i+1}$. Since $\widetilde{X}_q^{i+1}$ is the solution of (\ref{eq:loc_prob_g}), $X^{i+1}$ must be a feasible point of (\ref{eq:prob_g}), which follows from (\ref{eq:loc_glob}). As this holds for all $i\geq 0$, $X^i$ is then also a feasible point of (\ref{eq:prob_g}), i.e., $X^i\in\mathcal{S}$, if $i>0$. Since $X^i=C_q^i\widetilde{X}_q^i$, and using (\ref{eq:loc_glob}), we find that $\widetilde{X}_q^i$ as defined in (\ref{eq:X_fixed}) is a feasible point of (\ref{eq:loc_prob_g}), i.e., $\widetilde{X}_q^i\in\widetilde{\mathcal{S}}_q^i$.
\end{proof}

The results of Lemma \ref{lem:X_in_constraints} mean that all points $(X^{i})_{i>0}$ generated by the algorithm will be in the constraint set of the global problem (\ref{eq:prob_g}). Additionally, the final result states that $\widetilde{X}_q^i$ itself is a feasible point of (\ref{eq:loc_prob_g}), which is important to achieve convergence and a monotonic decrease in $f$, as it allows the algorithm to stay in the current point $X^i$ if no reduction in $f$ can be obtained at node $q$ in iteration $i$, in which case $X^{i+1}=X^i$ (a formal proof is given in \cite{musluoglu2022unifiedp2}).

\begin{rem}
  The fact that the local problem (\ref{eq:loc_prob_g}) inherits the structure from the global problem (\ref{eq:prob_g}) is one of the key differences between the DASF framework and the nonlinear Gauss-Seidel method (sometimes referred to as the alternating optimization method), which would consist of only updating $X_q^i$ in (\ref{eq:prob_g}), while freezing the other $X_k^i$'s $\forall k\neq q$. In the latter case, the subproblem that has to be solved in each iteration typically has a different structure than the original one, often leading to problems that are more difficult to solve or for which a straightforward solver might not even exist. Moreover, the extra degrees of freedom to manipulate the $X_k$'s of other nodes through the $G_k$ matrices in the parameterization (\ref{eq:X_param}) allow to optimize $X^i$ over a larger subset of $\mathbb{R}^{M\times Q}$ in each individual iteration, leading to larger descents in the loss function $f$ in each iteration. In Gauss-Seidel methods, these $G_k$ matrices do not exist, i.e., all $X_k$'s are fixed except one.
\end{rem}

\begin{rem}\label{rem:data_sink}
  By combining (\ref{eq:compress_y}) and (\ref{eq:X_param}), we find that
  \begin{equation}
    \widetilde{X}_q^{(i+1)T}\widetilde{\mathbf{y}}_q^i(t)=X^{(i+1)T}\mathbf{y}(t),
  \end{equation}
  which means that node $q$ always has access to the filtered signal $X^{(i+1)T}\mathbf{y}=\widehat{\mathbf{y}}^{i+1}$ based on the most recent version of the filter $X^{i+1}$. If any of the other nodes would act as a data sink, the updating node $q$ has to forward the observations of $\widehat{\mathbf{y}}^{i+1}$ to the data sink.
\end{rem}

\subsection{Multiple Signals, Deterministic Terms and Variables}\label{sec:mult_var}
The DASF algorithm can be immediately adapted to the generalized version of (\ref{eq:prob_g}) defined in (\ref{eq:prob_g_full}), i.e., with multiple filters, signals and deterministic terms. In the case of multiple signals (stochastic variables) or deterministic terms appearing in the problem in the forms $X^T\mathbf{y}$ and $X^TB$ respectively, every single object is treated as previously presented, creating new data to be communicated between the nodes for every additional expression. Examples include the GEVD and TRO given in Table \ref{tab:ex_prob} having two signals $\mathbf{y}$ and $\mathbf{v}$. An example of a problem with two deterministic terms is given in Section \ref{sec:ex_sim}.

Similarly, we could also consider cases with multiple optimization variables. Taking the example of CCA given in Table \ref{tab:ex_prob}, the two optimization variables $(X,W)$ appear as $X^T\mathbf{y}$ and $W^T\mathbf{v}$ in the problem. Then, nodes $k\neq q$ compress their signals as $\widehat{\mathbf{y}}_k^i=X_k^{iT}\mathbf{y}_k$ and $\widehat{\mathbf{v}}_k^i=W_k^{iT}\mathbf{v}_k$ and transmit them to node $q$. Then, node $q$ solves its local problem as
\begin{equation}
  \left(\widetilde{X}_q^{i+1},\widetilde{W}_q^{i+1}\right)=\underset{(\widetilde{X}_q,\widetilde{W}_q)\in\widetilde{\mathcal{S}}_q^i}{\text{argmin }}\varphi\left(\widetilde{X}_q^T\widetilde{\mathbf{y}}_q^i(t),\widetilde{W}_q^T\widetilde{\mathbf{v}}_q^i(t)\right).
\end{equation}
Partitioning $\widetilde{W}_q$ as
\begin{equation}
  \widetilde{W}_q=[W_q^T,H_1^T,\dots,H_{q-1}^T,H_{q+1}^T,\dots,H_K^T]^T,
\end{equation}
similarly to (\ref{eq:X_tilde_part}) for $\widetilde{X}_q$, node $q$ sends the $G_k^{i+1}$ and $H_k^{i+1}$'s to corresponding nodes $k$ such that they update their local variables as in (\ref{eq:full_upd}). The same procedure can be applied for more than two variables.

\begin{rem}\label{rem:gamma}
  The communication burden required to transmit the compressed terms $X_k^{iT}B_k$ is minimal because $B_k$'s are deterministic parameters, as opposed to the signals $\mathbf{y}_k$, which require sending batches of multiple compressed observations in each iteration to estimate the signal statistics at the updating node (see Remark \ref{rem:batch}). The communication cost of the deterministic part can be further reduced when we have an expression of the form $(X^TB^{(1)})\cdot (X^TB^{(2)})^T=X^T\Gamma X$, where $\Gamma$ is a block-diagonal deterministic matrix written as
  \begin{equation}\label{eq:Gamma_blkdiag_def}
    \Gamma=\textit{BlkDiag}(\Gamma_1,\dots,\Gamma_K),
  \end{equation}
  where each $\Gamma_k$ is known by node $k$. Each node could then transmit $X_k^{iT}\Gamma_kX_k^i\in\mathbb{R}^{Q\times Q}$ at iteration $i$ instead of $X_k^{iT}B_k^{(1)}$ and $X_k^{iT}B_k^{(2)}$, which is more efficient when $Q<2L$. Although expression (\ref{eq:Gamma_blkdiag_def}) is quite specific, it is encountered often in spatial filtering, for example for orthogonality constraints such as $X^TX=I_Q$ or $\ell_2$-norm regularization terms. For example, consider the PCA constraint $X^TX=I_Q$, where $B^{(1)}=B^{(2)}=I_M$. Then it is sufficient for the nodes to transmit $X_k^{iT}X_k^i$ instead of $X_k^i$.
\end{rem}

\section{Topology-Independent DASF (TI-DASF)}\label{sec:ti_dsf}
Until this point, we have considered fully-connected WSNs only, where every node in the network is a neighbor of every other node. In this section, we extend our discussions and describe the DASF algorithm for other network topologies. For this purpose, we first consider star topologies which are helpful to introduce the main idea, which will then lead to generalizations to tree topologies and finally to any (connected) network topology.

\subsection{Star Topologies}
We keep the same definitions introduced in Sections \ref{sec:prob_state} and \ref{sec:fc_dsf} and consider now that the WSN has a central node $c\in\mathcal{K}$ to which every other node $k\neq c$ is connected (having only node $c$ as a neighbor). In the case where the center node $c$ is the updating node, we have the same setting as the fully-connected case and the steps described in the previous section apply, therefore we present here a strategy when the updating node $q\neq c$. A straightforward approach would be to let node $c$ relay all the data from all other nodes to create a virtually fully-connected network. However, this would put high bandwidth requirements on node $c$, which would not scale well with respect to network size. Instead, we claim that it is sufficient for node $q$ to have access to the signal defined in (\ref{eq:y_Xy}), which is slightly rewritten here as
\begin{equation}\label{eq:sum_Xy}
  \widehat{\mathbf{y}}^i(t)=X^{iT}\mathbf{y}(t)=X_q^{iT}\mathbf{y}_q(t)+\sum_{k\neq q}\widehat{\mathbf{y}}_k^i(t).
\end{equation}
Note that the second term can be computed at node $c$ (which includes node $c$'s own sensor observations $\mathbf{y}_c$, as well as the compressed signals $\widehat{\mathbf{y}}_k^i$ of the nodes $k\neq q$) such that only the sum has to be forwarded instead of the individual terms. The data received by node $q$ from $c$ is then a $Q-$channel signal given by
\begin{equation}
  \widehat{\mathbf{y}}^i_{c\rightarrow q}(t)\triangleq \sum_{k\in\mathcal{K}\backslash\{q\}}\widehat{\mathbf{y}}^i_k(t),
\end{equation}
and a similar expression can be written for the deterministic terms:
\begin{equation}
  \widehat{B}^i_{c\rightarrow q}\triangleq \sum_{k\in\mathcal{K}\backslash\{q\}}\widehat{B}^i_k.
\end{equation}

From the perspective of node $q$, the network consists of only itself and node $c$, so by following analogous steps to Algorithm \ref{alg:fc_dsfo}, node $q$ creates its vector of locally available data:
\begin{equation}\label{eq:star_data}
  \begin{aligned}
    \widetilde{\mathbf{y}}^i_q(t)&=[\mathbf{y}_q^T(t),\widehat{\mathbf{y}}_{c\rightarrow q}^{iT}(t)]^T,\\
    \widetilde{B}^i_q&=[B_q^T,\widehat{B}_{c\rightarrow q}^{iT}]^T,
  \end{aligned}
\end{equation}
such that the corresponding $\widetilde{X}_q^{i+1}$ is obtained at iteration $i$ by solving the local problem (\ref{eq:loc_prob_g}) using the data that is available at node $q$, as in (\ref{eq:star_data}). Similar to (\ref{eq:X_tilde_i1}), we define the partitioning:
\begin{equation}
  \widetilde{X}^{i+1}_q=[X_q^{(i+1)T},G_{c}^{(i+1)T}]^T\in\mathbb{R}^{(M_q+Q)\times Q},
\end{equation}
where $G^{i+1}_{c}$ is analogous to the $G_k^{i+1}$'s in the previous section. Since node $q$ has only one link, there is only one such matrix resulting from solving the local problem (\ref{eq:loc_prob_g}), which is sent to node $c$. Finally, the central node $c$ disseminates this matrix $G_{c}^{i+1}$ to the other nodes to update their local compressor as in (\ref{eq:full_upd}), but with $G_k=G_c$ for all $k$.

\subsection{Tree Topologies}
We now consider a network represented by a tree, i.e., a graph without cyclic paths. A leaf node is defined as a node with a single neighbor, i.e., a node $k$ for which $|\mathcal{N}_k|=1$. Recalling that our objective is to be able to recreate (\ref{eq:sum_Xy}) at the updating node $q$, we perform an in-network fusion across the different tree branches that are rooted in node $q$. This fusion can be done in a bottom-up fashion without central coordination. Indeed, the strategy for each node $k\neq q$ is to wait until it has received the compressed signals from all its neighbors except one (denoted as node $n$), sum these and add its own compressed signal $\widehat{\mathbf{y}}_k^i$, and transmit to its remaining neighbor $n$. Formally, the compressed signal being sent from node $k\neq q$ to $n$ at iteration $i$ is
\begin{equation}\label{eq:sum_fwd}
  \widehat{\mathbf{y}}_{k \rightarrow n}^i(t)=X_k^{iT}\mathbf{y}_k(t)+\sum_{l\in\mathcal{N}_k\backslash\{n\}}\widehat{\mathbf{y}}_{l\rightarrow k}^i(t).
\end{equation}
Note that this is a recursive definition, which is bootstrapped by the leaf nodes, for which the second (recursive) term vanishes. This data fusion flow is illustrated in Figure \ref{fig:tree_diagram} for an example network. The fused signals will eventually arrive at the updating node $q$ which receives
\begin{equation}\label{eq:sum_fwd_n}
  \widehat{\mathbf{y}}_{n\rightarrow q}^i(t)=X_n^{iT}\mathbf{y}_n(t)+\sum_{k\in\mathcal{N}_n\backslash\{q\}}\widehat{\mathbf{y}}_{k\rightarrow n}^i(t)=\sum_{k\in\mathcal{B}_{nq}}\widehat{\mathbf{y}}_k^i(t),
\end{equation}
from each of its neighbors $n\in\mathcal{N}_q$, where we define $\mathcal{B}_{nq}$ to be the connected subgraph containing node $n$ when the link between nodes $n$ and $q$ is removed (see Figure \ref{fig:tree_diagram}). The same process is applied to the deterministic terms such that node $q$ receives
\begin{equation}\label{eq:sum_fwd_B}
  \widehat{B}_{n\rightarrow q}^i=X_n^{iT}B_n+\sum_{k\in\mathcal{N}_n\backslash\{q\}}\widehat{B}_{k\rightarrow n}^i=\sum_{k\in\mathcal{B}_{nq}}\widehat{B}_k^i
\end{equation}
from all its neighbors $n\in\mathcal{N}_q$.

\begin{figure}[t]
  \includegraphics[width=0.48\textwidth]{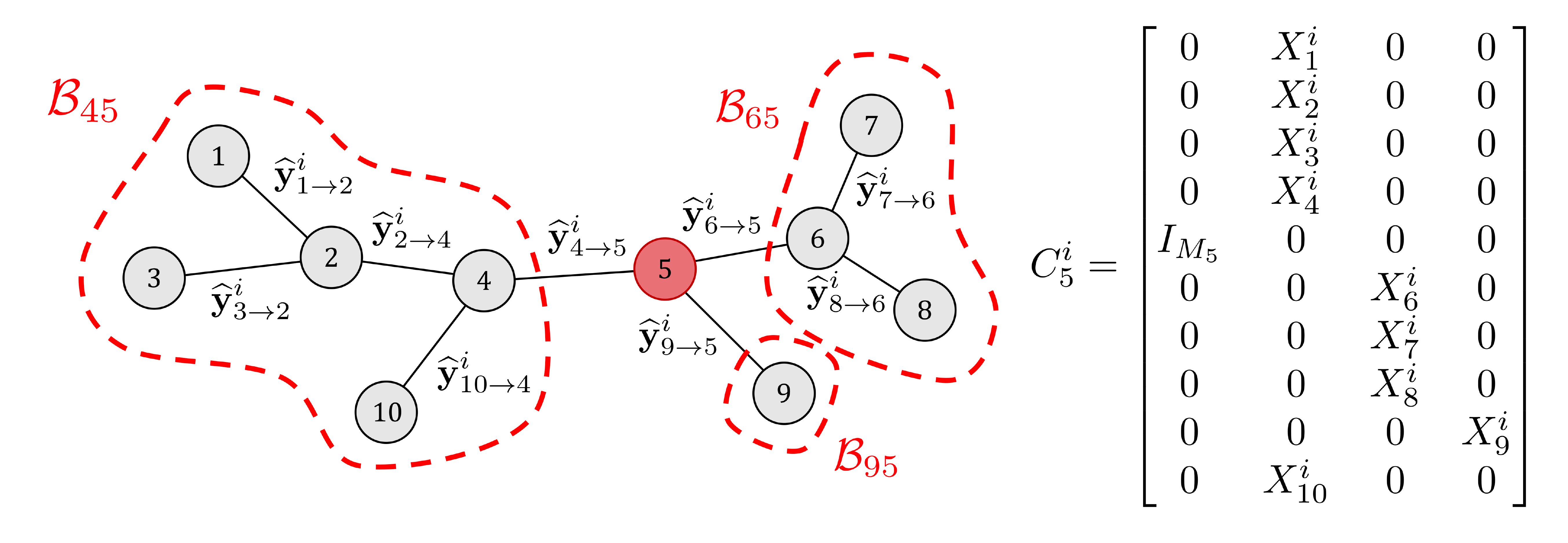}
  \caption{\cite{musluoglu2021distributed} Example of a tree network where the updating node is node $5$. Each neighbor of node $5$ creates its own cluster containing the nodes ``hidden'' from node $5$ behind them, shown here as $\mathcal{B}_{45}$, $\mathcal{B}_{65}$, $\mathcal{B}_{95}$. The resulting transition matrix is given by $C_5^i$.}
  \label{fig:tree_diagram}
\end{figure}

Writing $\mathcal{N}_q=\{n_1,\dots,n_{|\mathcal{N}_q|}\}$, we have the vector of available data at the updating node $q$:
\begin{equation}\label{eq:tree_data}
  \begin{aligned}
    \widetilde{\mathbf{y}}_q^i(t)&=[\mathbf{y}_q^T(t),\widehat{\mathbf{y}}_{n_1\rightarrow q}^{iT}(t),\dots,\widehat{\mathbf{y}}_{n_{|\mathcal{N}_q|}\rightarrow q}^{iT}(t)]^T,\\
    \widetilde{B}_q^i&=[B_q^T,\widehat{B}_{n_1\rightarrow q}^{iT},\dots,\widehat{B}_{n_{|\mathcal{N}_q|}\rightarrow q}^{iT}]^T.
  \end{aligned}
\end{equation}
We note that the relationship between the network-wide data $\mathbf{y}$, $B$ and the locally available $\widetilde{\mathbf{y}}_q^i$, $\widetilde{B}_q^i$ can again be described by means of a compression matrix $C_q^i$ as in (\ref{eq:compress_y})-(\ref{eq:compress_B}), such that $\widetilde{\mathbf{y}}_q^i=C_q^{iT}\mathbf{y}$ and $\widetilde{B}_q^i=C_q^{iT}B$. An example of such a matrix $C_q^i$ is shown in Figure \ref{fig:tree_diagram}. We recommend the reader to use the example of this figure to appreciate the structure of this matrix, yet we also provide a general definition for completeness. In general, this matrix can be defined as
\begin{equation}\label{eq:cqi_tree}
  C_q^i=\left[
    \begin{array}{c|c}
    0 &  \\
    I_{M_q} & \Theta_{-q}^i \\
    0 & 
    \end{array}
    \right],
\end{equation}
where $I_{M_q}$ is placed in the $q-$th block-row, and $\Theta_{-q}^i$ is a block matrix with $K$ block-rows and $|\mathcal{N}_q|$ block-columns, where the block at the $k-$th block-row and $m-$th block-column is represented by $\Theta_{-q}^i(k,m)\in\mathbb{R}^{M_k\times Q}$. Each block-column corresponds to one of the neighbors $n\in\mathcal{N}_q$ of $q$, which we re-index as $m_n\in\{1,\dots,|\mathcal{N}_q|\}$. Then, we have
\begin{equation}\label{eq:Theta_qi_def}
  \Theta_{-q}^i(k,m_n)=\begin{cases}
    X_k^i & \text{if $k\in\mathcal{B}_{nq}$} \\
    0 & \text{otherwise}
    \end{cases}.
\end{equation}
As in the previous cases, the transition from the local variable to the network-wide one is given by $X=C_q^i\widetilde{X}_q$ at node $q$ and iteration $i$. We can verify that $\widetilde{X}_q$ is a feasible point of the local problem if and only if $C_q^i\widetilde{X}_q$ is a feasible point of the global problem, i.e., (\ref{eq:loc_glob}) and more generally Lemma \ref{lem:X_in_constraints} also holds here. Node $q$ then solves its local problem (\ref{eq:loc_prob_g}) using the locally available data described in (\ref{eq:tree_data}), to obtain $\widetilde{X}_q^{i+1}$, partitioned as
\begin{equation}\label{eq:X_tilde_tree}
  \widetilde{X}_q^{i+1}=\Big[X_q^{(i+1)T},G_{n_1}^{(i+1)T},\dots,G_{n_{|\mathcal{N}_q|}}^{(i+1)T}\Big]^T.
\end{equation}
Each $G_{n}^{i+1}$ is then disseminated into the corresponding subgraph $\mathcal{B}_{nq}$ through node $n$ (and the nodes behind it in $\mathcal{B}_{nq}$) and every node updates its compressor as
\begin{equation}\label{eq:upd_tree}
  X_k^{i+1}=\begin{cases}
  X_q^{i+1} & \text{if $k=q$} \\
  X_k^{i}G_n^{i+1} & \text{if $k\in\mathcal{B}_{nq}$, $n\in\mathcal{N}_q$}
  \end{cases}
\end{equation}
such that we again have $X^{(i+1)T}\mathbf{y}=\widetilde{X}_q^{(i+1)T}\widetilde{\mathbf{y}}_q^i$ and $X^{(i+1)T}B=\widetilde{X}_q^{(i+1)T}\widetilde{B}_q^i$.

From (\ref{eq:upd_tree}), we observe that we have parameterized the global variable $X$ at node $q$ and iteration $i$ as
\begin{equation}\label{eq:X_param_tree}
  X=C_q^i\widetilde{X}_q=\begin{bmatrix} X_1^i\boxed{G_{n(1)}}\\
    \vdots\\
    X_{q-1}^i\boxed{G_{n(q-1)}}\\
    \boxed{X_q}\\
    X^i_{q+1}\boxed{G_{n(q+1)}}\\
    \vdots\\
    X^i_{K}\boxed{G_{n(K)}}
  \end{bmatrix},
\end{equation}
where $C_q^i$ is given in (\ref{eq:cqi_tree}) and $n(k)$ is the neighbor $n\in\mathcal{N}_q$ of node $q$ such that $k\in\mathcal{B}_{nq}$. This can be compared with (\ref{eq:X_param}) for the fully-connected case. Since the optimization variable of (\ref{eq:loc_prob_g}) is partitioned as in (\ref{eq:X_tilde_tree}), we can see from (\ref{eq:X_param_tree}) that, similarly to the fully-connected case, the updating node $q$ can ``freely'' update its filter $X_q$, while the filters of the other nodes can only change up to a right-hand side matrix multiplication with a $G_n-$matrix (the updating variables are the framed variables in (\ref{eq:X_param_tree})). In contrast to the fully-connected case in (\ref{eq:X_param}), some of the $G_n$'s are constrained to be equal due to the network topology since $k,l\in\mathcal{B}_{nq}\Rightarrow G_{n(k)}=G_{n(l)}$. This implies that the degrees of freedom in the updating steps of the tree-based DASF algorithm are determined by the number of neighbors of the updating node.

\begin{rem}
  Sometimes it is possible that a node cannot generate $Q$ linearly independent output channels using (\ref{eq:sum_fwd}). For example, this happens if node $k$ is a leaf node and $M_k<Q$. In this case, node $k$ will send its raw uncompressed sensor data $\mathbf{y}_k$ instead of $\widehat{\mathbf{y}}_{k\rightarrow n}^i$ defined in (\ref{eq:sum_fwd}). The neighbor $n$ that receives the raw data from node $k$ will then treat this data as part of its own sensor signals, i.e., $\mathbf{y}_n$ is stacked with $\mathbf{y}_k$ and its sensor channel count becomes $M_n+M_k$. In other words, the data flow starting at the leaf nodes can initially consist of raw sensor channels until a node has more than $Q$ channels to compress them into a $Q-$channel signal. Therefore, the number of transmitted channels is at most $Q$ per node, but can also be less than $Q$. An analogous statement applies to the deterministic terms $B_k$.
\end{rem}

\begin{figure}[!t]
  \removelatexerror
  \begin{algorithm}[H]
    \DontPrintSemicolon
  \caption{Topology-Independent Distributed Adaptive Signal Fusion (TI-DASF) Algorithm\;
  Code available in \cite{musluoglu2022dsfotoolbox}}\label{alg:ti_dsfo}
  \SetKwInOut{Output}{output}
  \Output{$X^*$}
  \BlankLine
  Initialize $X^0$, $i\gets0$.\;
  \Repeat
  {
  Choose the updating node as $q\gets (i\mod K)+1$.\;
  1) The network $\mathcal{G}$ is pruned into a tree $\mathcal{T}^i(\mathcal{G},q)$.\;

  2) Every node $k$ collects a new batch of $N$ samples of $\mathbf{y}_k$. All nodes compress these to $N$ samples of $\widehat{\mathbf{y}}^{i}_k$ as in (\ref{eq:y_hat}). $\widehat{B}_k^i$ is computed using (\ref{eq:B_hat}).\;
  3) The nodes sum-and-forward their compressed data towards node $q$ via the recursive rule (\ref{eq:sum_fwd}) (and a similar rule for the $\widehat{B}_k^i$'s). Node $q$ eventually receives $N$ samples of $\widehat{\mathbf{y}}^{i}_{n\rightarrow q}$ along with $\widehat{B}_{n\rightarrow q}^i$ given in (\ref{eq:sum_fwd_n})-(\ref{eq:sum_fwd_B}) from all its neighbors $n\in\mathcal{N}_q$.\; 

  \At{Node q}
  {
    4a) Compute $\widetilde{X}_q^{i+1}$ as the solution of (\ref{eq:loc_prob_g}) where $\widetilde{\mathbf{y}}_q^i$, $\widetilde{B}_q^i$ and $\widetilde{X}_q$ are redefined as in (\ref{eq:tree_data}) and (\ref{eq:X_tilde_tree}). If the solution of (\ref{eq:loc_prob_g}) is not unique, select the solution which minimizes $||\widetilde{X}_q^{i+1}-\widetilde{X}_q^i||_F$ with $\widetilde{X}_q^i$ defined as in (\ref{eq:X_fixed}).\;

    4b) Partition $\widetilde{X}^{i+1}_q$ as in (\ref{eq:X_tilde_tree}).\;
    4c) Disseminate $G_n^{i+1}$ to all nodes in $\mathcal{B}_{nq}$, $\forall n\in\mathcal{N}_q$.\;
  }
  
  5) Every node updates $X_k^{i+1}$ according to (\ref{eq:upd_tree}).\;
  
  $i\gets i+1$\;
  }
  \end{algorithm}
  \small{\textbf{Note:} Each iteration uses a different batch of $N$ samples in step 2, i.e., the iterations can be spread over different time segments in order to avoid retransmitting the same batch of $N$ samples across the network (see also Remark \ref{rem:batch}). This makes that the sample time index $t$ is coupled to the iteration index $i$.\\
  \textbf{Note:} As in FC-DASF, the fused output signal $X^T\mathbf{y}(t)=\widehat{\mathbf{y}}(t)$ can be computed as $\widetilde{X}_q^{(i+1)T}\widetilde{\mathbf{y}}_q^i(t)$ at node $q$ without extra transmissions (see also Remark \ref{rem:data_sink}).}
  \hrule
\end{figure}

\subsection{General Connected Graphs}
Suppose the network is represented by a connected graph $\mathcal{G}$, which can potentially contain cycles. The main difference with the previous subsection is that there is more than one choice to forward the compressed data to the updating node $q$. We therefore propose to prune the graph $\mathcal{G}$ into a (different) tree at each iteration $i$, so that we can apply the same steps as the ones described for the tree topology case. The resulting tree is denoted as $\mathcal{T}^i(\mathcal{G},q)$ to highlight the fact that the pruning function $\mathcal{T}^i$ depends on the current updating node $q$. An example of the data flow in pruned networks for two different updating nodes is given in Figure \ref{fig:data_flow}.

The choice of the mapping function $\mathcal{T}^i$ is a free design choice as long as the resulting tree remains connected. However, the convergence results in \cite{musluoglu2022unifiedp2} that define the bound (\ref{eq:upper_bound}) also assume that the pruning function does not remove the links between the updating node $q$ and its neighbors. Indeed, if (\ref{eq:upper_bound}) is satisfied, then this rule ensures that $J\leq (1+|\mathcal{N}_q|)Q^2$ at any updating node $q$, which is one of the convergence conditions for the DASF algorithm in \cite{musluoglu2022unifiedp2}. Furthermore, disconnecting node $q$ from a neighbor would also reduce the convergence speed, since it would lead to a smaller number of resulting $G_n$'s in (\ref{eq:X_tilde_tree})-(\ref{eq:upd_tree}) and therefore a lower number of degrees of freedom in the local optimization problem (\ref{eq:loc_prob_g}), typically leading to a smaller descent of the cost $\varphi$ or $f$. In comparison, in the case of fully-connected networks, we were able to use $(K-1)$ of such $G_n$'s in each iteration, which --- as we will show in Section \ref{sec:ex_sim} --- leads to the fastest convergence.

A simple distributed protocol to establish $\mathcal{T}^i$, i.e., to set up a tree that satisfies the aforementioned rule, is to let the updating node $q$ broadcast a token to each neighbor. Once a node receives a token, it acknowledges a link to the node from which it received it, and it then broadcasts that token to its remaining neighbors from which it did not receive a token. This process continues until all nodes have received a token. If a node receives multiple tokens, it only connects to the parent node from which it received the token first (in case of a tie, it makes a random choice). 

\begin{figure}[t]
  \includegraphics[width=0.48\textwidth]{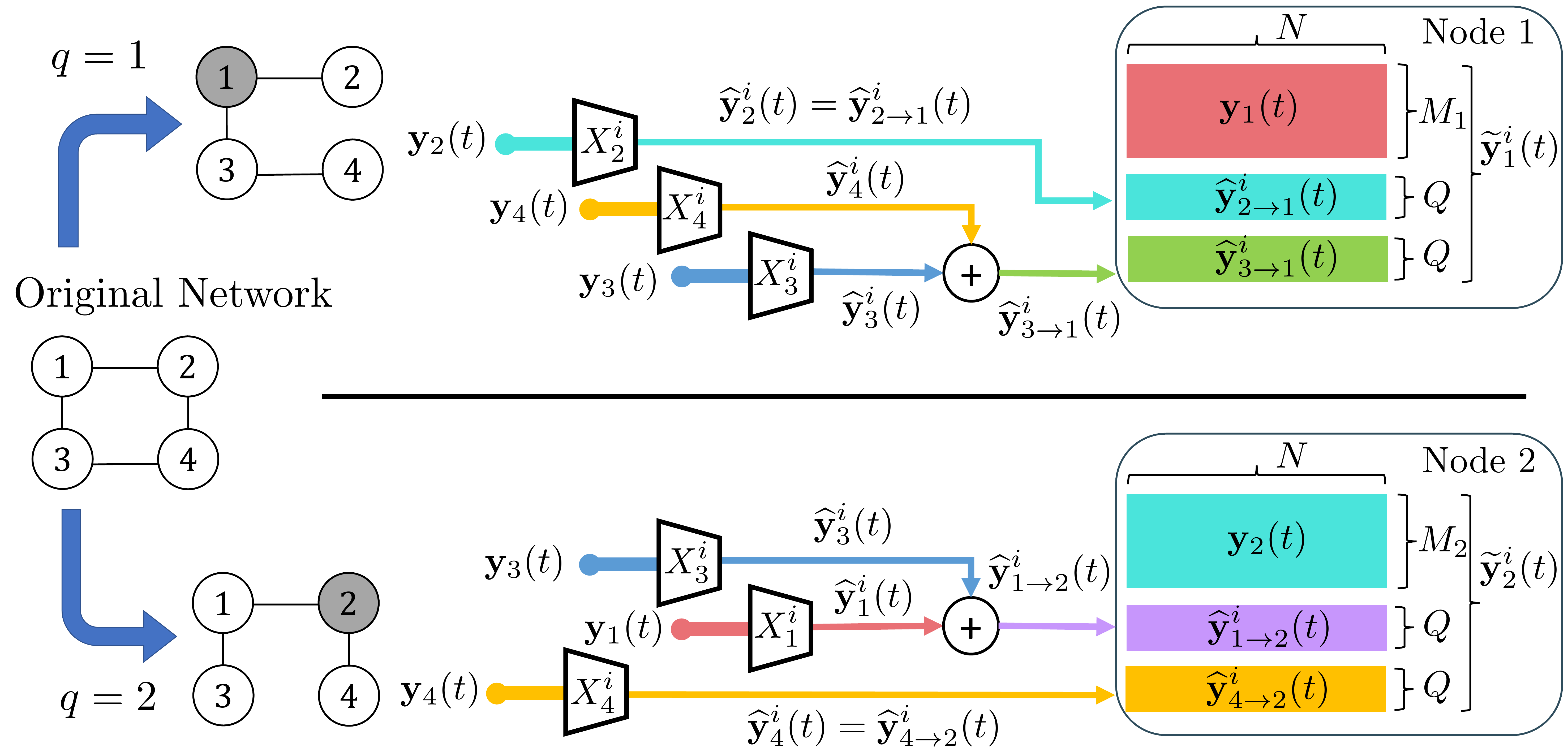}
  \caption{Data flow for a $4-$node network when nodes $q=1$ (top) and $q=2$ (bottom) are the updating node. The updating node has always access to its own data $\mathbf{y}_q$, while receiving fused signals from every other node in the network through its neighboring nodes.}
  \label{fig:data_flow}
\end{figure}

The full TI-DASF algorithm is given in Algorithm \ref{alg:ti_dsfo} and convergence analyses are provided in \cite{musluoglu2022unifiedp2}. We note that FC-DASF (Algorithm \ref{alg:fc_dsfo}) is a special case of TI-DASF (Algorithm \ref{alg:ti_dsfo}). The same generalizations as explained in Section \ref{sec:mult_var} apply for TI-DASF as well.

\begin{rem}
  In the TI-DASF algorithm, each node $k$ transmits only $NQ$ samples of $\mathbf{y}_k$ per iteration, which is independent of the number of neighbors or the total number of nodes in the network. As opposed to the fully-connected case, the TI-DASF algorithm can reduce the total communication burden, even when $Q\geq M_k$. This is because naively relaying all the raw sensor data to a fusion center node for centralized processing would require most nodes to send more than $NQ$ samples as they would also have to forward the data from their neighbors, and their neighbors' neighbors, etc (see Figure \ref{fig:data_flow_tree}). Obviously, such a relaying approach does not scale well with the network size, whereas the per-node bandwidth requirements in the TI-DASF algorithm are independent of the number of nodes.
\end{rem}

\section{Examples and Simulations}\label{sec:ex_sim}
In this Section, we present examples of problems that fit the DASF framework and demonstrate the performance of the algorithm in various settings to gain insights into the convergence behavior as a function of $Q$, $K$ and the topology. The examples also serve as illustrations to familiarize the reader with how to recognize SFO problems of the form (\ref{eq:prob_g}) or (\ref{eq:prob_g_full}), and how to translate these into the DASF framework. It is noted that several existing distributed algorithms can be shown to be special cases of the proposed unified DASF framework (Table \ref{tab:ex_prob}). Since these special cases have been validated already, and to show the generalizing properties of the framework, we will validate it on a few new problems that fit our framework. For this purpose, we have also published a companion toolbox \cite{musluoglu2022dsfotoolbox} which allows to automatically generate and simulate a distributed algorithm for any arbitrary problem of the form (\ref{eq:prob_g}) or (\ref{eq:prob_g_full}). The only requirement is that the user provides a solver for the centralized problem (\ref{eq:prob_g}) or (\ref{eq:prob_g_full}). The software will use this solver to compute the updating step in each iteration of the FC- or TI-DASF algorithm, as the update requires solving a compressed local instance of (\ref{eq:prob_g}) or (\ref{eq:prob_g_full}).

\begin{figure}[t]
  \includegraphics[width=0.48\textwidth]{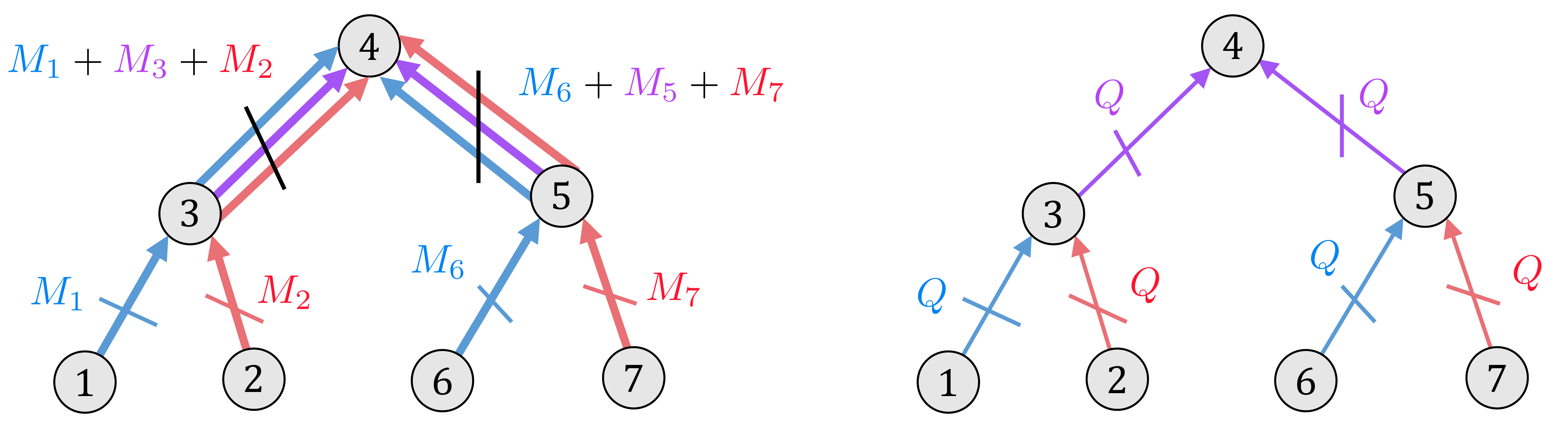}
  \caption{Comparison between a straightforward relaying approach (left) and the scalable fuse-and-forward approach the DASF algorithm uses (right). In this example, node $4$ is the updating node.}
  \label{fig:data_flow_tree}
\end{figure}

In our experiments, we refer to \textit{randomly generated trees} as trees where each node has between $0$ and $4$ children with $1.7$ children on average\footnote{The number of children nodes is selected randomly from $[0,1,2,3,4]$, which is distributed following the probability vector $[0.2,0.3,0.2,0.2,0.1]$.}. We consider two different sensor signals $\mathbf{y}$ and $\mathbf{v}$ measured at each node throughout this section, following the mixture model given by
\begin{align}
  \mathbf{y}(t)&=\Pi_s\cdot \mathbf{s}(t)+\mathbf{n}(t),\label{eq:signal_model_v}\\
  \mathbf{v}(t)&=\Pi_r\cdot \mathbf{r}(t)+\mathbf{y}(t)\nonumber \\
  &=\Pi_r\cdot \mathbf{r}(t)+\Pi_s\cdot \mathbf{s}(t)+\mathbf{n}(t),
\end{align}
with $\mathbf{r}(t)$, $\mathbf{s}(t)\overset{i.i.d.}{\sim}\mathcal{N}(0,\sigma_r^2)$, $\mathbf{n}(t)\overset{i.i.d.}{\sim}\mathcal{N}(0,\sigma_n^2)$ for every entry and time instance $t$. In the experimental settings of Sections \ref{sec:qcqp} to \ref{sec:scqp}, the entries of $\Pi_s$ and $\Pi_r$ are independent of the time $t$ and drawn from the uniform distribution within the interval $[-0.5,0.5]$. In Section \ref{sec:adaptivity}, we will consider an adaptive setting where $\Pi_s$ is time-dependent. We assume that both $\mathbf{y}$ and $\mathbf{v}$ are observable at the nodes (this is possible, e.g., if the source $\mathbf{r}$ has an on-off behavior).
In all the simulations, we take the number of samples of the signals to be communicated between the nodes to be $N=10^4$ and each node has an equal number of channels $M_k=M/K$ (where the total number of channels $M$ and the total number of nodes $K$ will vary). The convergence of the DASF algorithm is assessed by tracking the normalized error $\epsilon$:
\begin{equation}\label{eq:epsilon_error}
  \epsilon(X^i)=\frac{||X^i-X^*||_F^2}{||X^*||_F^2}.
\end{equation}
The optimal value $X^*$ is computed by solving the problems we present in the following paragraphs using centralized solvers. In case the centralized problem has multiple possible solutions, i.e., $|\mathcal{X}^*|>1$, we select $X^*\in\mathcal{X}^*$ in (\ref{eq:epsilon_error}) that best matches $X^i$ in the final iteration of the simulation. This resolves the ambiguity of the solution, while the plots would still reveal non-convergence in case the algorithm would arrive in a limit cycle that switches between multiple accumulation points, i.e., we would observe subsequences of $(X^i)_i$ converging to different solutions of the problem. The parameters chosen for each following problem in Sections \ref{sec:qcqp} to \ref{sec:scqp} are summarized in Table \ref{tab:table_param}. In these experiments, we aim to observe the theoretical convergence result of the DASF algorithm and therefore consider stationary and ergodic signals. On the other hand, Section \ref{sec:adaptivity} has a slightly different experimental setting as we aim to demonstrate the adaptive properties of the DASF algorithm, in which case stationarity does not hold.

\begin{table}[!t]
  \renewcommand{\arraystretch}{1.5}
  \caption{Summary of parameters used in the simulations.}
  \label{tab:table_param}
  \centering
  \begin{tabularx}{0.48\textwidth}{ >{\centering\arraybackslash}c | >{\centering\arraybackslash}X | >{\centering\arraybackslash}X | >{\centering\arraybackslash}X }
  \hline
  Experiment&  Section \ref{sec:qcqp}  &  Section \ref{sec:tro} & Section \ref{sec:scqp}  \\ \hhline{=|=|=|=}
  $Q$ &  $3$   &  $5$ & $\{1,3,5,7\}$ \\ \hline
  $K$ &  $\{10,25,50\}$   &   \multicolumn{2}{c}{$30$} \\ \hline
  $M$, $M_k$ & \multicolumn{3}{c}{$M=450$, $M_k=M/K$, $\forall k\in\mathcal{K}$} \\ \hline
  Signal Statistics & $\sigma_r^2=\sigma_n^2=1$   &   $\sigma_r^2=0.5$, $\sigma_n^2=0.1$ & $\sigma_r^2=\sigma_n^2=1$ \\ \hline
  $N$ & \multicolumn{3}{c}{$10000$} \\ \hline
  Monte Carlo Runs & \multicolumn{3}{c}{$100$} \\ \hline
  \end{tabularx}
\end{table}

\subsection{Quadratically Constrained Quadratic Problem}\label{sec:qcqp}
In this subsection, we will solve the following problem:
\begin{equation}\label{eq:qcqp}
  \begin{aligned}
    \underset{X}{\text{minimize } } \quad & \frac{1}{2}\mathbb{E}[||X^T\mathbf{y}(t)||^2]-\text{tr}(X^TA)\\
    \textrm{subject to} \quad & \text{tr}(X^T X)\leq \alpha^2,\;X^T\mathbf{c}=\mathbf{d},\\
    \end{aligned}
\end{equation}
where we take $\sigma_n^2=\sigma_r^2=1$ for the noise and signal variance. In this problem, we have three deterministic inner products of the form $X^TB$ where $B$ is known a priori. Two of them are the terms $X^TA$ and $X^T\mathbf{c}$ where $A\in\mathbb{R}^{M\times Q}$ and $\mathbf{c}\in\mathbb{R}^{M}$. The third one comes from $X^TX$, which can be written as $(X^TI_M)\cdot(X^TI_M)^T$ which reveals the term $X^TB$ with $B=I_M$ (see also Remark \ref{rem:gamma}). The values of $\alpha\in\mathbb{R}$ and $\mathbf{d}\in\mathbb{R}^{Q}$ have been chosen randomly while ensuring that $\alpha^2\geq||\mathbf{d}||^2/||\mathbf{c}||^2$ which would otherwise make the problem infeasible. We can write $\mathbb{E}[||X^T\mathbf{y}(t)||^2]=\text{tr}(X^TR_{\mathbf{yy}}X)$ where $R_{\mathbf{yy}}=\mathbb{E}[\mathbf{y}(t)\mathbf{y}^T(t)]$ is the covariance matrix of the signal $\mathbf{y}$. We note that the matrix $R_{\mathbf{yy}}$ is assumed to be unknown, as the signal statistics are to be learned by the algorithm, so it should not be seen as a deterministic term. Then, the local problem at iteration $i$ that node $q$ needs to solve is
\begin{equation}\label{eq:qcqp_loc}
  \begin{aligned}
    \underset{\widetilde{X}_q}{\text{minimize } } \quad & \frac{1}{2}\text{tr}(\widetilde{X}_q^TR^i_{\widetilde{\mathbf{y}}_q\widetilde{\mathbf{y}}_q}\widetilde{X}_q)-\text{tr}(\widetilde{X}_q^T\widetilde{A}_q^i)\\
    \textrm{subject to} \quad & \text{tr}(\widetilde{X}_q^TC_q^{iT}C_q^i\widetilde{X}_q)\leq \alpha^2,\;\widetilde{X}_q^T\widetilde{\mathbf{c}}_q^i=\mathbf{d},
    \end{aligned}
\end{equation}
where $R^i_{\widetilde{\mathbf{y}}_q\widetilde{\mathbf{y}}_q}=\mathbb{E}[\widetilde{\mathbf{y}}_q^i(t)\widetilde{\mathbf{y}}_q^{iT}(t)]$ is the correlation matrix of the locally available signal $\widetilde{\mathbf{y}}_q^i$ at node $q$ and iteration $i$. Note that the compression matrix $C_q^i$ appears in the quadratic constraints, which is indeed what one obtains when computing $\widetilde{B}_q^i$ with $B=I_M$, which directly follows from (\ref{eq:compress_B}).

\begin{figure}[t]
  \includegraphics[width=0.48\textwidth]{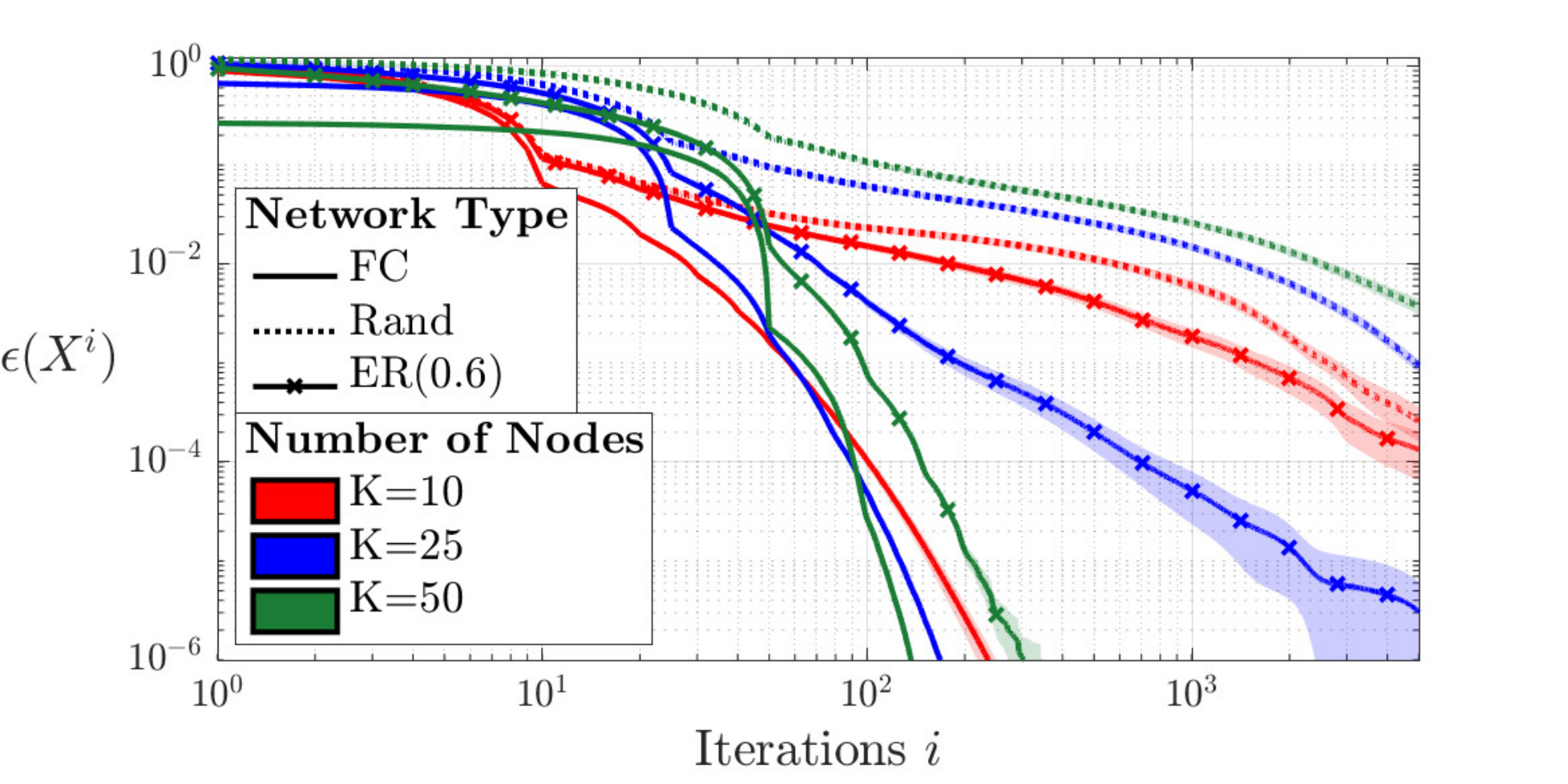}
  \caption{Convergence comparison of the DASF algorithm solving (\ref{eq:qcqp}) in fully-connected networks (\textit{FC}), randomly generated trees (\textit{Rand}) and Erd\H{o}s-R\'{e}nyi random graphs with connection probability of $0.6$ (\textit{ER(0.6)}) for various network sizes. The bold lines represent the mean values across $100$ Monte Carlo runs, while the shaded areas delimit the standard error of the mean around them.}
  \label{fig:K_plot}
\end{figure}

For this experiment, we take the number of channels to be $M=200$, take $Q=3$ and look at the behavior of the algorithm for a varying number of nodes in the network with $K\in\{10,25,50\}$ (the number of channels per node $M_k=M/K$ therefore changes for each $K$). The results are shown in Figure \ref{fig:K_plot}. For the fully-connected case, we see that the smaller networks, i.e., when $K$ is small, converge faster in the first iterations. This is because they are able to do a full round update (over all nodes) in a smaller number of iterations compared to larger networks. However, the larger networks can eventually ``catch up'' and even surpass the convergence speed of their counterparts with smaller $K$ after a certain number of iterations due to the larger number of degrees of freedom (i.e., a larger number of $G_k$ matrices in each iteration). This property is however not observed for randomly generated trees as here the number of $G_k$ matrices at each node is related to its number of neighbors, hence independent of the network size. Overall, the fully-connected topologies lead to faster convergence and the randomly generated trees are the slowest, which is consistent with the expectations based on the amount of degrees of freedom (i.e., the number of $G_k$ matrices in each update). Networks for which the topology is neither a tree nor fully-connected are expected to fall in between these two extreme cases. We add that the high variance observed in Figure \ref{fig:K_plot} over various Monte-Carlo runs is due to the fact that Problem (\ref{eq:qcqp}) has many parameters chosen randomly.

\begin{figure}[t]
  \includegraphics[width=0.48\textwidth]{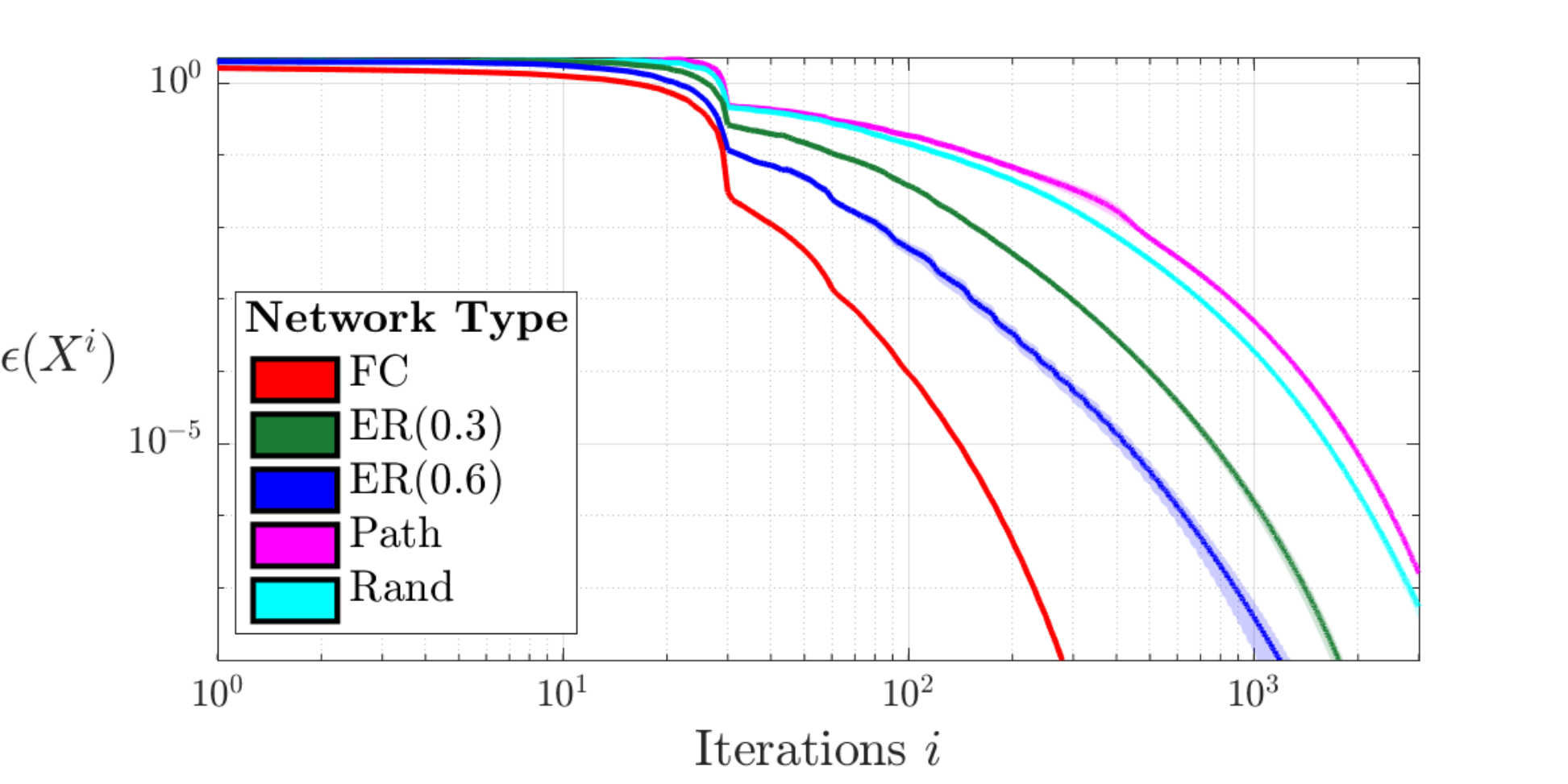}
  \caption{Convergence comparison of the DASF algorithm solving (\ref{eq:tro}) for various network settings namely fully-connected networks (\textit{FC}), Erd\H{o}s-R\'enyi random graphs with connection probability $p$ (\textit{ER($p$)}), path or line graphs (\textit{Path}) and randomly generated trees (\textit{Rand}). The bold lines represent the mean values across $100$ Monte Carlo runs, while the shaded areas delimit the standard error of the mean around them.}
  \label{fig:topo_plot}
\end{figure}

\subsection{The Trace Ratio Optimization Problem}\label{sec:tro}
We now consider the problem:
\begin{equation}\label{eq:tro}
  \begin{aligned}
    \underset{X}{\text{maximize } } \quad & \frac{\mathbb{E}[||X^T\mathbf{v}(t)||^2]}{\mathbb{E}[||X^T\mathbf{y}(t)||^2]}=\frac{\text{tr}(X^TR_{\mathbf{vv}}X)}{\text{tr}(X^TR_{\mathbf{yy}}X)}\\
    \textrm{subject to} \quad & X^T X=I_Q,\\
    \end{aligned}
\end{equation}
often referred to as the trace ratio or trace quotient optimization (TRO) problem \cite{wang2007trace}. It is noted that a distributed algorithm for this TRO problem has been proposed in \cite{musluoglu2021distributed}  where techniques tailored to the TRO problem have been used. Applying the generic Algorithm \ref{alg:ti_dsfo} to Problem (\ref{eq:tro}) results in an alternative distributed algorithm for the TRO problem, based on the DASF framework. $R_{\mathbf{yy}}$ and $R_{\mathbf{vv}}$ are again spatial correlation matrices of $\mathbf{y}$ and $\mathbf{v}$ respectively, assumed to be unknown to the DASF algorithm. We take the signal and noise variance to be $\sigma_r^2=0.5$ and $\sigma_n^2=0.1$ respectively. Problem (\ref{eq:tro}) can be solved using the solver in \cite{wang2007trace} by iteratively computing generalized eigenvalue decompositions.

At updating node $q$ and iteration $i$, (\ref{eq:tro}) is translated to
\begin{equation}\label{eq:dtro}
  \begin{aligned}
    \underset{\widetilde{X}_q}{\text{maximize } } \quad & \frac{\text{tr}(\widetilde{X}_q^TR^i_{\widetilde{\mathbf{v}}_q\widetilde{\mathbf{v}}_q}\widetilde{X}_q)}{\text{tr}(\widetilde{X}_q^TR^i_{\widetilde{\mathbf{y}}_q\widetilde{\mathbf{y}}_q}\widetilde{X}_q)}\\
    \textrm{subject to} \quad & \widetilde{X}_q^TC_q^{iT}C_q^i\widetilde{X}_q=I_Q,\\
    \end{aligned}
\end{equation}
where $R^i_{\widetilde{\mathbf{y}}_q\widetilde{\mathbf{y}}_q}$ and $R^i_{\widetilde{\mathbf{v}}_q\widetilde{\mathbf{v}}_q}$ are estimated in the same way as in the previous example. Therefore, the solver in \cite{wang2007trace} can be applied to solve the local problem (\ref{eq:dtro}) in step 4a of Algorithm \ref{alg:ti_dsfo}, by replacing $R_{\mathbf{yy}}$ and $R_{\mathbf{vv}}$ by $R^i_{\widetilde{\mathbf{y}}_q\widetilde{\mathbf{y}}_q}$ and $R^i_{\widetilde{\mathbf{v}}_q\widetilde{\mathbf{v}}_q}$ respectively, where $q$ is the updating node at iteration $i$.

The experimental results are shown in Figure \ref{fig:topo_plot}, where we highlight the differences in convergence depending on the topology of the network. In particular, we look at fully-connected networks, random trees, line topologies (trees where each node has two neighbors, except leaf nodes which have only one neighbor), and Erd\H{o}s-R\'enyi (ER) models, generated using \cite{perraudin2014gspbox}. In each case, we keep the number of nodes, the number of channels and the compression dimension $Q$ constant, with $K=30$, $M=450$ and $Q=5$ respectively. We observe that the more the network is connected the faster the DASF algorithm converges, the fastest being the fully-connected networks, while the slowest convergence is obtained for line graphs. This is again in line with the results and discussion on the degrees of freedom in Section \ref{sec:qcqp}.

\subsection{Quadratic Problem with a Spherical Constraint}\label{sec:scqp}
Let us now consider:
\begin{equation}\label{eq:scqp}
  \begin{aligned}
    \underset{X}{\text{minimize } } \quad & \frac{1}{2}\mathbb{E}[||X^T\mathbf{y}(t)||^2]+\text{tr}(X^TA)\\
    \textrm{subject to} \quad & \text{tr}(X^T X)=1,\\
    \end{aligned}
\end{equation}
where $\sigma_n^2=\sigma_r^2=1$ and the elements of $A$ have been chosen independently at random. Note that this problem differs from (\ref{eq:qcqp}) in the sense that it has a non-convex constraint set due to the non-linear equality constraint. The local problem at node $q$ and iteration $i$ can be written as
\begin{equation}\label{eq:dscqp}
  \begin{aligned}
    \underset{\widetilde{X}_q}{\text{minimize } } \quad & \frac{1}{2}\text{tr}(\widetilde{X}_q^TR^i_{\widetilde{\mathbf{y}}_q\widetilde{\mathbf{y}}_q}\widetilde{X}_q)+\text{tr}(\widetilde{X}_q^T\widetilde{A}_q^i)\\
    \textrm{subject to} \quad & \text{tr}(\widetilde{X}_q^TC_q^{iT}C_q^i\widetilde{X}_q)=1.\\
    \end{aligned}
\end{equation}
For this experiment, we fix $K=30$, $M=450$ and consider various number of filters $Q\in\{1,3,5,7\}$. We observe in Figure \ref{fig:Q_plot} that the larger the number of filters $Q$, the faster the algorithm converges. This is expected as a larger value for $Q$ implies that more information can be used at every node, however at the expense of a larger communication cost. As for the previous examples, we see that the DASF framework converges faster for fully-connected networks compared to randomly generated trees for Problem (\ref{eq:scqp}).

\begin{figure}[t]
  \includegraphics[width=0.48\textwidth]{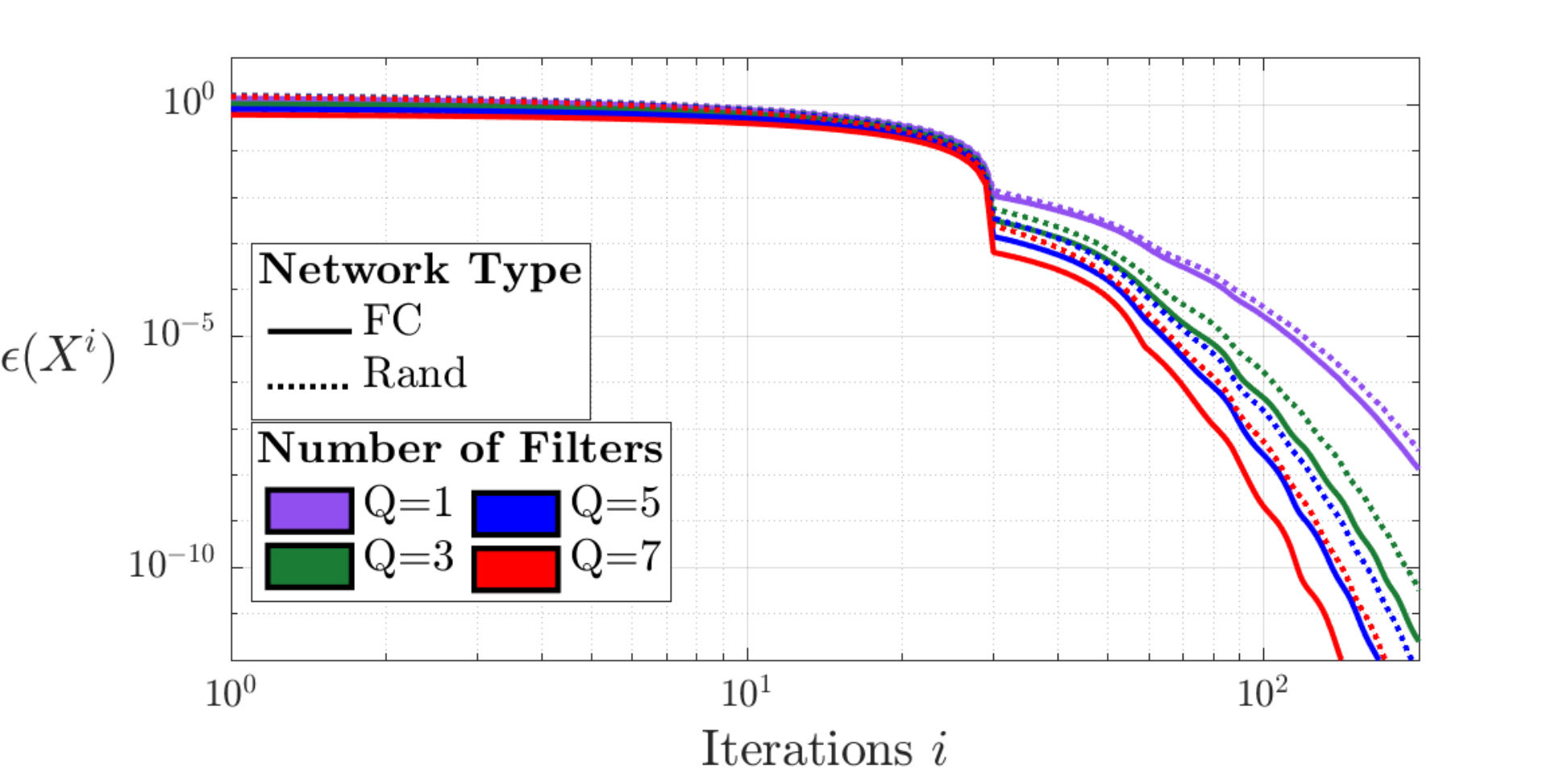}
  \caption{Convergence comparison of the DASF algorithm solving (\ref{eq:scqp}) for various number of filters $Q$ for fully-connected networks (\textit{FC}) and randomly generated trees (\textit{Rand}). The bold lines represent the mean values across $100$ Monte Carlo runs, while the shaded areas delimit the standard error of the mean around them.}
  \label{fig:Q_plot}
\end{figure}

\subsection{DASF in a Tracking Problem}\label{sec:adaptivity}

In this final experiment, we consider the MMSE problem
\begin{equation}\label{eq:prob_ls}
  \underset{\mathbf{x}\in\mathbb{R}^{M}}{\text{minimize } } \quad \mathbb{E}[|s(t)-\mathbf{x}^T\mathbf{y}(t)|^2],
\end{equation}
where $s$ is a one-dimensional signal, implying $Q=1$. The problem is solved using the DASF algorithm on a randomly generated Erd\H{o}s-R\'enyi network with connection probability $0.8$. The network contains $K=10$ nodes, each measuring a signal $\mathbf{y}_k$ with $M_k=4$ channels. At each iteration $i$, $N=10^4$ new times samples of $\mathbf{y}_k$ are used, such that $i=\lfloor t/N \rfloor$. The network-wide signal is given by
\begin{equation}
  \mathbf{y}(t)=\mathbf{p}(t) \cdot s(t)+\mathbf{n}(t),
\end{equation}
for each $t$, with $s(t)\overset{i.i.d.}{\sim}\mathcal{N}(0,1)$ and $\mathbf{n}(t)\overset{i.i.d.}{\sim}\mathcal{N}(0,0.1)$ for every entry and time instance $t$. The main difference with (\ref{eq:signal_model_v}) is that now the steering vector $\mathbf{p}$ (corresponding to the mixture matrix $\Pi_s$ in (\ref{eq:signal_model_v})) changes with time, implying that the statistical properties of $\mathbf{y}$ are time-dependent as well. At each time instance $t$, the solution of (\ref{eq:prob_ls}) is given by $\mathbf{x}^*(t)=(R_{\mathbf{yy}}(t))^{-1}\mathbf{r}_{\mathbf{y}s}(t)$, where $R_{\mathbf{yy}}(t)=\mathbb{E}[\mathbf{y}(t)\mathbf{y}^T(t)]$ and $\mathbf{r}_{\mathbf{y}s}(t)=\mathbb{E}[\mathbf{y}(t)s(t)]$. For each $t$, we have $\mathbf{p}(t)=(1-\lambda(t))\cdot \mathbf{p}_0+\lambda(t)\cdot (\mathbf{p}_0+\Delta)$, where $\mathbf{p}_0$ and $\Delta$ are time-independent vectors of which the entries are drawn from $\mathcal{N}(0,1)$ and $\mathcal{N}(0,0.01)$ respectively. The function $\lambda$ used in this experiment is represented in Figure \ref{fig:adaptivity_plot}, where it can be observed that smooth, as well as abrupt, changes are modeled. At each iteration, the ground-truth solution of (\ref{eq:prob_ls}) is estimated as
\begin{equation}
  \mathbf{x}^{*i}=(R_{\mathbf{yy}}^i)^{-1}\mathbf{r}_{\mathbf{y}s}^i,
\end{equation}
where
\begin{equation}
  R_{\mathbf{yy}}^i=\sum_{\tau=iN}^{iN+N-1}\mathbf{y}(\tau)\mathbf{y}^T(\tau),\; \mathbf{r}_{\mathbf{y}s}^i=\sum_{\tau=iN}^{iN+N-1}\mathbf{y}(\tau)s(\tau).
\end{equation}

\begin{figure}[t]
  \includegraphics[width=0.48\textwidth]{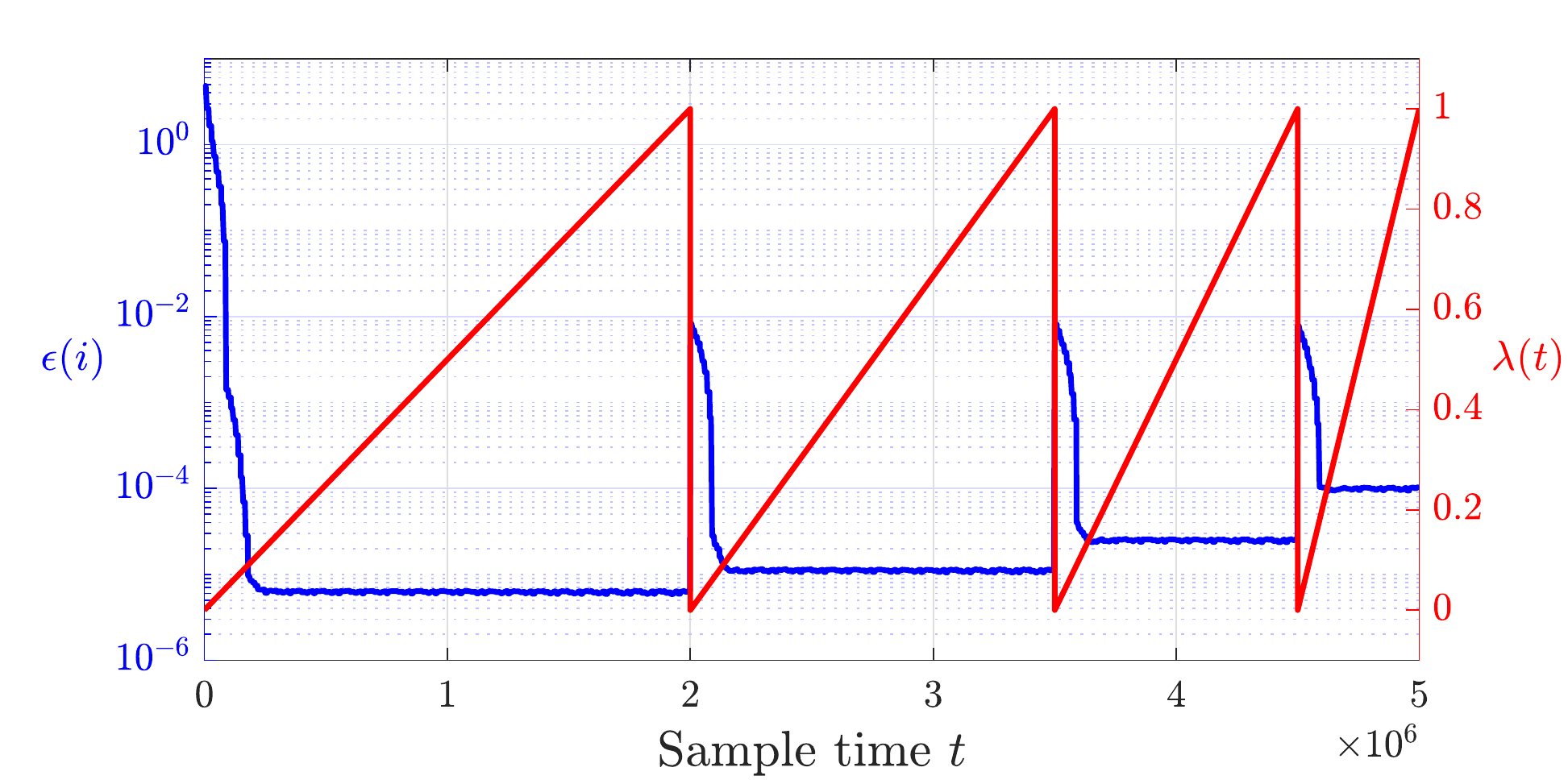}
  \caption{Error $\epsilon$ over time of the DASF algorithm in an adaptive setting (blue). The signal statistics change over time and depend on the function $\lambda$, represented in red. The relationship between the time $t$ and the iterations $i$ is given by $i=\lfloor t/N\rfloor$.}
  \label{fig:adaptivity_plot}
\end{figure}

Figure \ref{fig:adaptivity_plot} shows the median value over $100$ Monte Carlo runs of the error $\epsilon$:
\begin{equation}
  \epsilon(i)=\frac{||\mathbf{x}^i-\mathbf{x}^{*i}||^2}{||\mathbf{x}^{*i}||^2}
\end{equation}
over time, where $i=\lfloor t/N\rfloor$. The algorithm can adapt to abrupt changes in signal statistics, i.e., when the value of $\lambda$ suddenly changes, which translates to an initial jump in the error followed by a decrease. We observe that the DASF algorithm is also able to track slow changes in the statistical properties of the signal, shown by constant error values $\epsilon$ over the iterations, despite changes in the value of $\lambda$. Since the optimal solution $\mathbf{x}^{*i}$ changes at each iteration, the error $\epsilon$ settles around a certain threshold, which is higher for larger rates of change in $\lambda$ and due to the approximation error.

\section{Conclusion}
In this paper, we have proposed the DASF framework which contains a large number of well-known distributed spatial filter design problems and algorithms as a special case. For intelligibility purposes, we have first addressed the case of fully-connected networks (FC-DASF) and then generalized it to any network represented by a connected graph (TI-DASF). In order to reduce the communication burden, the nodes only communicate compressed data across the network. An interesting property of the resulting distributed algorithm is that the local problem to be solved at an updating node has the same structure as the original network-wide centralized problem, such that the same solver can be used. The convergence properties of the algorithm have been illustrated by means of four example instances of the (D)SFO problem (\ref{eq:prob_g}) or its more general form in (\ref{eq:prob_g_full}), one of the examples focusing on the adaptive properties of the DASF algorithm. A formal convergence analysis with convergence and optimality proofs, including examples on how the convergence conditions can be checked, are provided in a companion paper \cite{musluoglu2022unifiedp2}. We have also provided a toolbox to automatically design and test the DASF algorithm for any user-defined instances of the (D)SFO problem (\ref{eq:prob_g}) or (\ref{eq:prob_g_full}), available in \cite{musluoglu2022dsfotoolbox}.

\section{Acknowledgments}
The authors would like to thank Charles Hovine for the brainstorming sessions and help with Sections \ref{sec:practical_E} and \ref{sec:adaptivity}.




\bibliographystyle{IEEEtran}
\bibliography{IEEEabrv,IEEEexample}

\end{document}